\documentclass[pra,reprint]{revtex4-1}
\usepackage{graphicx} 
\usepackage{epstopdf}

\usepackage{array} 
\usepackage{verbatim} 
\usepackage{amsmath,amsfonts,amssymb,amscd}
\usepackage{amsthm}
\usepackage{tabularx}
\usepackage{stmaryrd}
\usepackage{enumerate}
\usepackage[ruled]{algorithm2e}

\usepackage{hyperref}
\hypersetup{
    bookmarksnumbered=true, 
    unicode=false, 
    pdfstartview={FitH}, 
    pdftitle={}, 
    pdfauthor={}, 
    pdfsubject={}, 
    pdfcreator={}, 
    pdfproducer={}, 
    pdfkeywords={}, 
    pdfnewwindow=true, 
    colorlinks=true, 
    linkcolor=blue, 
    citecolor=blue, 
    filecolor=blue, 
    urlcolor=blue 
}
\usepackage{subfigure}

\theoremstyle{plain}
\newtheorem{thm}{Theorem}

\theoremstyle{definition}

\newcommand{\eq}[1]{(\hyperref[eq:#1]{\ref*{eq:#1}})}

\renewcommand{\sec}[1]{\hyperref[sec:#1]{Section~\ref*{sec:#1}}}
\newcommand{\thrm}[1]{\hyperref[thm:#1]{Theorem~\ref*{thm:#1}}}
\newcommand{\lemm}[1]{\hyperref[lemm:#1]{Lemma~\ref*{lemm:#1}}}
\newcommand{\prop}[1]{\hyperref[prop:#1]{Proposition~\ref*{prop:#1}}}
\newcommand{\corr}[1]{\hyperref[corr:#1]{Corollary~\ref*{corr:#1}}}
\newcommand{\fig}[1]{\hyperref[fig:#1]{\ref*{fig:#1}}}
\newcommand{\tbl}[1]{\hyperref[tbl:#1]{\ref*{tbl:#1}}}

\newcommand{\oo}{\mathcal{O}}

\usepackage[usenames,dvipsnames]{xcolor}
\usepackage{xcolor}

\DeclareMathAlphabet{\matheu}{U}{eus}{m}{n}

\newcommand{\I}{{\mathbb I}}

\newcolumntype{L}[1]{>{\raggedright}p{#1}}
\newcolumntype{C}[1]{>{\centering}p{#1}}
\newcolumntype{R}[1]{>{\raggedleft}p{#1}}
\newcolumntype{D}{>{\centering\arraybackslash}X}

\input{Qcircuit}

\begin{document}
\title{Error rates and resource overheads of encoded three-qubit gates}
\author{Ryuji Takagi, Theodore J. Yoder and Isaac L. Chuang}
\affiliation{Department of Physics, Massachusetts Institute of Technology, 77 Massachusetts Avenue, Cambridge, Massachusetts 02139, USA}
\date{\today}

\begin{abstract}
A non-Clifford gate is required for universal quantum computation, and, typically, this is the most error-prone and resource intensive logical operation on an error-correcting code. Small, single-qubit rotations are popular choices for this non-Clifford gate, but certain three-qubit gates, such as Toffoli or controlled-controlled-$Z$ (CCZ), are equivalent options that are also more suited for implementing some quantum algorithms, for instance, those with coherent classical subroutines. 
Here, we calculate error rates and resource overheads for implementing logical CCZ with pieceable fault-tolerance, a non-transversal method for implementing logical gates. 
We provide a comparison with a non-local magic-state scheme on a concatenated code and a local magic-state scheme on the surface code. We find the pieceable fault-tolerance scheme particularly advantaged over magic states on concatenated codes and in certain regimes over magic states on the surface code.
Our results suggest that pieceable fault-tolerance is a promising candidate for fault-tolerance in a near-future quantum computer.
\end{abstract}

\maketitle

\section{Introduction}
Quantum error-correcting codes are the most promising route to scalable quantum computation. However, some of their limitations are well-known. For instance, a major problem is that a single code cannot support a full set of universal, transversal operations \cite{Zeng2007,Chen2008,Eastin2009a}. Often, and always for 2D designs \cite{Bravyi2013a}, the missing gate is not in the normalizer of the Pauli group; that is, it is non-Clifford. 

The techniques of gate-teleportation \cite{Gottesman1999b} and magic-states \cite{Bravyi2005a} can overcome the lack of a non-Clifford gate. Different magic-states can be created to implement small $Z$-rotations such as the $T$-gate or 3-qubit operations, like Toffoli or controlled-controlled-$Z$ (CCZ). 
However, the process to create a magic-state occurs post-selectively and recursively and leads to large resource overheads. Although improving consistently \cite{Jones2013a,Campbell2017a} approaching believed fundamental limits \cite{Bravyi2012a}, large resource demands remain a serious obstacle for near-future architectures.

Certain other approaches exist in the literature for implementing a universal gate-set while circumventing the use of magic-states. A popular approach is gauge-fixing \cite{Paetznick2013a,Anderson2014a,Bombin2015a}, in which a subsystem code can implement complementary sets of transversal logical gates depending on the settings of the gauge qubits. Another approach \cite{Jochym-OConnor2014a,Nikahd2016a,Nikahd2016b} concatenates different codes with complementary transversal gate sets to achieve the same effect in one larger code. Recently, this approach was shown to lead to asymptotic thresholds around $\sim10^{-3}$ albeit using more physical qubits than, for example, surface code magic-state distillation \cite{Chamberland2016d,Chamberland2017}.

Any fault-tolerant, universal computing scheme operating without magic states is expected to be a promising candidate for near-future architectures where fairly accurate physical components are supplied but space-time resources, like qubit count and circuit depth, are limited. 
The primary goal in this near-future regime is to achieve some desired target error rate after a finite-sized computation with small resource overheads.
Such constraints imply that the logical error rates of encoded gates and the first-level pseudothreshold \cite{Svore2005b} (called just pseudothreshold hereafter) are more important measures than asymptotic threshold, which only becomes meaningful with access to huge amounts of resources. 

To evaluate near-future fault-tolerant computation, we focus on another magic-less alternative that allows for a logical implementation of three-qubit gates, the pieceable fault-tolerance scheme \cite{Yoder2016c}. In this approach, a logical gate is done non-transversally through the ``round-robin" construction, and made fault-tolerant via partial error-correction performed throughout the circuit.
This construction has recently been used in \cite{Chao2017a} to perform fault-tolerant, universal computing on seven logical qubits requiring only four ancillary qubits and 15 code qubits. The circuit volume metric, a space-time resource measure that counts all gates weighted by the number of qubits involved, was used in \cite{Yoder2016c} to argue that pieceable fault-tolerance reduces logical gate overhead by nearly a factor of two over magic-state creation and injection. However, little was said about error rates of pieceable gates.

In this paper, we calculate these error rates and compare to magic-state schemes for implementing three-qubit non-Clifford gates. Our contenders are (1) a non-local magic-state scheme: magic-states created postselectively on Steane's 7-qubit code (also known as the smallest color code), (2) a local magic-state scheme: surface code magic-state distillation, and (3) pieceable fault-tolerance on the (a) 5-qubit \cite{Yoder2016c}, (b) 7-qubit \cite{Yoder2016c}, (c) $3\times3$ Bacon-Shor \cite{Yodera}, and (d) $3\times9$ Bacon-Shor \cite{Yodera} codes. Our metrics are (I) error rate of the logical gate and (II) circuit volume. Among concatenated schemes (1) and (3), we can definitively declare pieceable $3\times3$ Bacon-Shor the winner in both metrics (I) and (II). When comparing to (2), the picture is more complicated and interesting.
The pieceable $3\times3$ Bacon-Shor beats the surface code in error rate at low distance and in circuit volume when the physical error rate is sufficiently low compared with the desired target logical error. 
On the other hand, asymptotically in code distance, the surface code outperforms pieceable $3\times 3$ Bacon-Shor due to better scaling of logical error rate and volume with distance.

\section{Methods} 
We first describe our method to evaluate the logical error rates. Evaluating the surface code scheme (2) draws on the extensive literature on the topic \cite{Fowler2012a}. 
Our calculations of the logical error rates of schemes (1) and (3) at code distance $d=3$ are done by exact enumeration of all combinations of up to two faults in the circuit extended-rectangle (exREC) \cite{Aliferis2005a} under the standard depolarizing noise model (which serves as a model of average-case noise).
In \cite{Aliferis2005a}, a rigorous upper bound on the logical error rate under depolarizing noise is given.
In contrast, we provide formulas giving a rigorous lower bound as well as a tighter rigorous upper bound.
The lower and upper bounds on logical error rate also determine lower and upper bounds on the pseudothreshold. Having both bounds allows us to definitively prove a separation between two different schemes when it exists.
Our method also confers some advantages over a Monte Carlo simulation. First, we can rigorously verify our circuits are fault-tolerant under the chosen noise model by checking that all single faults are correctable. Second, once the counting is complete, we can independently vary noise for each type of gate.

Our standard noise breakdown assigns single-qubit gates, two-qubit gates, and three-qubit gates each their own failure probabilities $p_1,p_2,$ and $p_3$, respectively. 
In the circuit depolarization noise model, an $r$-qubit gate fails with one of the $4^r-1$ $r$-qubit Pauli errors with probability $p_r/(4^r-1)$. In principle, preparation and measurement could be treated separately as well, though we will assign them failure probabilities also equal to $p_1$. Bounds on the error rate can always be written as polynomials in $p_1,p_2,p_3$ as we discuss below.

Our ultimate goal in error-rate estimation is to find the probability the exREC is incorrect given that all ancillas pass verification. Denote this $P_{\text{fail}|\text{acc}}=\text{Pr}\left[\text{fail}|\text{acc}\right]$. Our counting gives the exact values of
\begin{align}
P^{(2)}_{\text{fail},\text{acc}}&=\text{Pr}\left[\text{fail},\text{acc},\le2\text{ faults}\right],\label{eq:fail_acc}\\
P^{(2)}_{\text{succ},\text{acc}}&=\text{Pr}\left[\neg\text{fail},\text{acc},\le2\text{ faults}\right],\label{eq:succ_acc}\\
P^{(2)}_{\text{rej}}&=\text{Pr}\left[\neg\text{acc},\le2\text{ faults}\right],\label{eq:rej}
\end{align}
as polynomials in $p_1,p_2,p_3$ with degree equal to the number of potentially faulty components in the entire exREC.  These exactly calculated quantities are enough to bound $P_{\text{fail},\text{acc}}=\text{Pr}\left[\text{fail},\text{acc}\right]$, $P_{\text{succ},\text{acc}}=\text{Pr}\left[\neg\text{fail},\text{acc}\right]$, and $P_{\text{acc}}=\text{Pr}\left[\text{acc}\right]$ as
\begin{align}
P^{(2)}_{\text{fail},\text{acc}}&\le P_{\text{fail},\text{acc}},\\
P^{(2)}_{\text{succ},\text{acc}}&\le P_{\text{succ},\text{acc}},\\
P_{\text{acc}} &\le 1-P^{(2)}_{\text{rej}}.
\end{align}
Thus,
\begin{equation}
\frac{P^{(2)}_{\text{fail},\text{acc}}}{1-P^{(2)}_{\text{rej}}}\le P_{\text{fail}|\text{acc}}=1-P_{\text{succ}|\text{acc}}\le1-\frac{P^{(2)}_{\text{succ},\text{acc}}}{1-P^{(2)}_{\text{rej}}}.
\label{eq:bound}
\end{equation}
 More details on the simulation including the description on how to obtain these polynomials can be found in Appendix~\ref{app:sim_details}.

Next, we consider evaluating the resource overhead.
There exist various resource measures such as qubit count, circuit volume, gate counts and so on. The number of reusable physical qubits is often taken as a physical resource measure in the literature.
However, it is not the best, especially when we would like to compare resource overheads between different codes, because there is ambiguity that comes with the level of parallelization we assume.
In this paper, we mainly focus on circuit volume, a space-time resource measure that counts all gates weighted by the number of qubits involved. 
Unlike physical qubit count, circuit volume takes into account the trade-off between space and time resources. The circuit volume is a space-time metric in the same vein as the ``quantum volume'' \cite{qvolume}, except for evaluating specific circuits rather than a universal quantum computer.

The circuit volume at a high concatenation level is easy to compute using the volume of the logical construction at the first level of encoding. 
Let $V^{(k)}_G$ be the volume for implementing circuit component $G$ at the $k^{\text{th}}$ level of concatenation. 
Then, there is a recursion relation between two concatenation levels, $V^{(k+1)}_G=\sum_{G'} N_G^{G'} V^{(k)}_{G'}$ where $N_G^{G'}$ is the number of the circuit component $G'$ in the logical construction of component $G$.
We can understand this as evolution of a vector of circuit volumes of each component via a transformation matrix determined by the logical gate constructions. Namely, we get
\begin{equation}
 {\bf V^{(k)}}=A^k {\bf V^{(0)}},
 \label{eq:higher_concate}
\end{equation}
where $A$ is the matrix $A_{ij}=N_{G_i}^{G_j}$, ${\bf V^{(k)}}_i = V^{(k)}_{G_i}$, and $V^{(0)}_G$ is the volume of an unencoded component. 
We set ${\bf V^{(k)}}=(V^{(k)}_3,V^{(k)}_2,V^{(k)}_1,V^{(k)}_{prep},V^{(k)}_{meas})^T$, where the components refer to the circuit volume of three qubit gates, two qubit gates, single qubit gates, $\ket{0}$ or $\ket{+}$ preparation, and measurement respectively. 
Note that $(V^{(0)}_3,V^{(0)}_2,V^{(0)}_1,V^{(0)}_{prep},V^{(0)}_{meas})=(3,2,1,1,1)$.

\section{Logical constructions}\label{circuit_details}
Here, we describe the logical constructions used in the simulation. Explicit descriptions of the circuits at the gate level can be found in \cite{circuit}. All of our constructions begin with a round of syndrome measurement and recovery (the leading error correction) and end with the same (the trailing error correction), in accordance with the exREC formalism \cite{Aliferis2005a}. The rest of the circuit may also include rounds of error correction, called intermediate, in accordance with pieceable fault-tolerance \cite{Yoder2016c}.
 
For the 5-qubit code, we implement a logical CCZ gate by the round-robin construction \cite{Yoder2016c} with three intermediate error corrections. 
The leading error correction and trailing error correction are done by Steane's error correction \cite{Steane1997}.
Since the 5-qubit code is non-CSS, a 10-qubit ancilla is needed to extract the entire syndrome simultaneously. 
We actually find that the circuit in \cite{Steane1997} needs some modification for non-CSS codes, which we discuss in the Appendix~\ref{app:generalized_steane} in detail.
For intermediate error corrections, we use Shor-type error correction with CAT states \cite{Shor2011a}.
The size of the CAT states is always four for measuring constant stabilizers (those that commute with the preceeding circuitry), but it varies for measuring non-constant stabilizers because their weight changes as they go through the CCZ gates. 
For our circuit, we need to use 9-CAT, 13-CAT, 9-CAT at maximum for the first, second and third intermediate error correction respectively. 

For the 7-qubit code, we consider the construction that requires only one intermediate error correction \cite{Yoder2016c}.
All of the error corrections are done by Steane's error correction. 
Since the 7-qubit code is a CSS code, correction of $Z$ type errors can be done separately from that of $X$ type errors, and only the encoded states $\ket{\bar{0}}$ and $\ket{\bar{+}}$ are needed. The state
$\ket{\bar{0}}$($\ket{\bar{+}}$) is verified by applying CNOT gates transversally to another noisy $\ket{\bar{0}}$($\ket{\bar{+}}$) and measuring it transversally (a Steane ancilla factory \cite{Steane1998a}). 
If some error is detected, we discard the state and start again.
For estimating the circuit volume, we consider a more resource-efficient state preparation method proposed by Goto \cite{Goto2016}.
Although we did not estimate the logical error rate using the Goto's method, we suspect that the change in the logical error rate between different verification methods would be small as indicated in \cite{Goto2016}.
Since intermediate $Z$-type error correction is not needed, we just apply the $X$-type error correction in the middle and notify the trailing error correction about possible locations of $Z$-type errors as described in \cite{Yoder2016c}. 

Logical CCZ on the Bacon-Shor code is implemented as proposed in \cite{Yodera}. On the $3\times 3$ Bacon-Shor we need no intermediate correction although we do use a non-Pauli recovery at the end. Furthermore, since the ancilla for the error correction is a tensor product of 3-CAT states, there is no need for verification since, modulo its stabilizers, an error on a 3-CAT is equivalent to a weight one error.
In contrast, the $3\times 9$ Bacon-Shor implements logical CCZ transversally, but it comes with a substantially larger overhead \cite{Yodera}.

For the non-local magic-state scheme, we use magic state injection on the 7-qubit code to implement a logical CCZ gate.
The CCZ magic state is defined by the stabilizers $\left< X_1\text{CZ}(2,3), X_2\text{CZ}(1,3), X_3\text{CZ}(1,2)\right>$.
The protocol consists of two parts, a state preparation circuit and a teleportation circuit. 
The state preparation starts with the +1 eigenstate of the second and the third stabilizer, $\ket{\bar{0}}\ket{\bar{+}}\ket{\bar{+}}$, and measures the first stabilizer \cite{Zhou2000}. 
Our circuit is a variant of the circuit in \cite{Monroe2014} which we modify to create the CCZ state instead. Two measurements of $X_1\text{CZ}(2,3)$ are done with complete error-correction in between. This makes the circuit fault-tolerant (to one fault).
If the two measurement results do not match, we discard the created state and start over again. If they match and they both show the result -1, we apply $\bar{Z}$ on the first code block to put it back to the desired magic state. If both show +1, we do not need to apply a correction. Like the pieceable 7-qubit case, all the error corrections are done using Steane's method \cite{Steane1997}.

 \section{Comparison of concatenated schemes} 
 We compute the logical error rates and resource overheads of pieceably fault-tolerant CCZ gates on the 5-qubit code, 7-qubit code \cite{Yoder2016c}, $3\times 3$ Bacon-Shor code and $3\times 9$ Bacon-Shor code \cite{Yodera}, and compare them to a magic-state scheme on the 7-qubit code.
Fig.~\fig{logical_error} shows the obtained logical error rates for these cases using two different settings of physical error rate, $p_1=p_2=p_3=p$ and $10p_1=p_2=0.1p_3=p$.
Lower and upper bounds on pseudothresholds are the crossing points of ``break-even" line and the upper and lower bounds for logical error rates. For both settings of physical error rate, the $3\times 3$ Bacon-Shor code has lower logical error rate than the magic-state scheme below pseudothreshold. 
For the 7-qubit code, whether the pieceable scheme has a lower rate than the magic-state scheme depends on the physical error rate setting.
The 5-qubit code has a large logical error rate due to a large number of pieces in the round-robin construction.
Similarly, the $3\times 9$ Bacon-Shor code has a higher logical error rate than the $3\times 3$ Bacon-Shor code because the size of the logical code block is obviously much bigger. Moreover, the $3\times 9$ Bacon-Shor needs to implement verification for 9-qubit CAT states.
 
We now compare the resource overheads. Table~\tbl{resource} shows the resource overheads to implement a logical CCZ gate with these constructions.
We assume that the ancillas are not reusable. 
Due to a finite ancilla verification rejection rate, the effective resource count is slightly higher than the values in the table. 
However, the rejection rate of the verification is $\oo(p)$, and the effective resource count is obtained by multiplying $(1-n_{rej}p)^{-1}$ where $n_{rej}$ is the number of error locations that lead to rejection. 
Since we are interested in the region $p<10^{-4}$, and the largest module involving verification is the magic state preparation circuit, which has $n_{rej}\sim 100$, increase in the resource due to verification is within 1\%.
Thus, it is safe to ignore the effects of verification. 
Besides using more 3-qubit gates, pieceable constructions on the 7-qubit and $3\times 3$ Bacon-Shor code have smaller resource overheads compared to the magic-state scheme. 
In particular, they have a significant reduction in circuit volume.
Fig.~\fig{volume} shows circuit volumes for the pieceable 7-qubit code, $3\times 3$ Bacon-Shor code, and magic-state scheme. 
Transformation matrices $A$ (see Eq.~\eqref{eq:higher_concate}) for these codes are given in Appendix~\ref{app:volume_calc}.

Combining the results for the logical error rates and circuit volume, we conclude that the pieceable construction on the $3\times 3$ Bacon-Shor code beats the magic-state scheme on the 7-qubit code in both the criteria. The pieceable construction on the 7-qubit code also beats magic state injection in circuit volume, and in logical error rate when $p_1=p_2=p_3$. 

\begin{figure}[htbp] 
\begin{center}
\subfigure[]{
\includegraphics[scale=0.25]{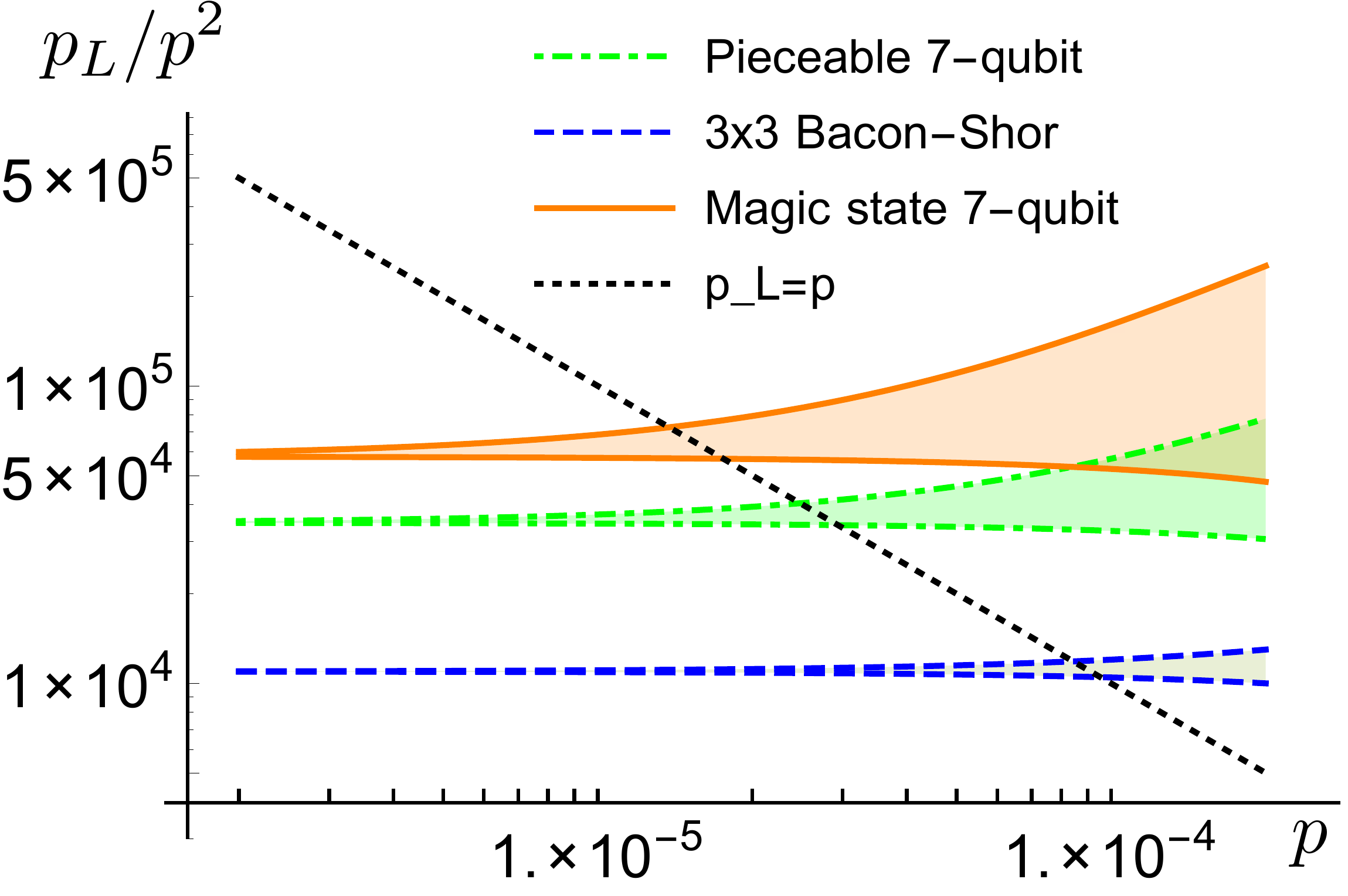}}
\subfigure[]{
\includegraphics[scale=0.25]{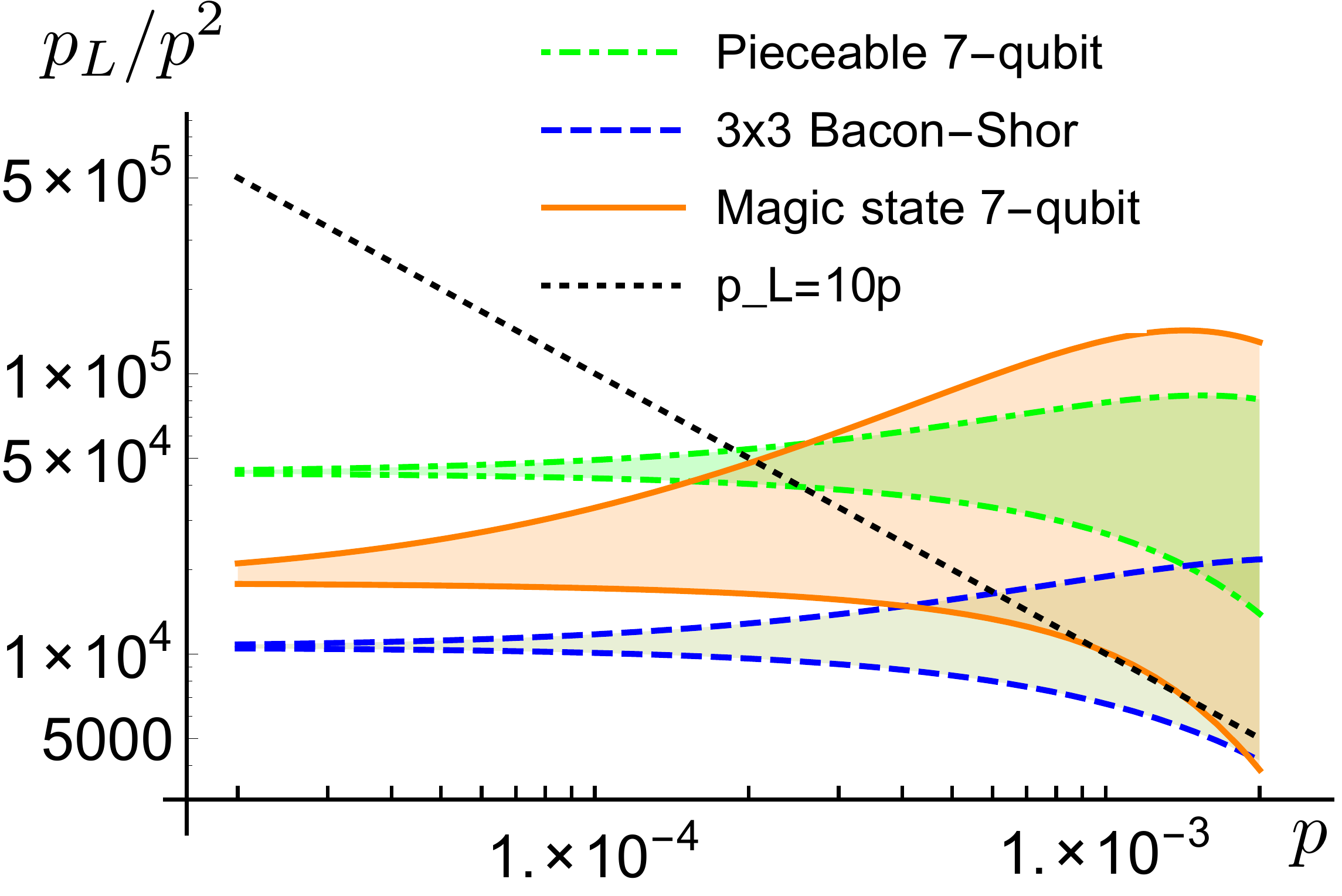}}
\subfigure[]{
\includegraphics[scale=0.25]{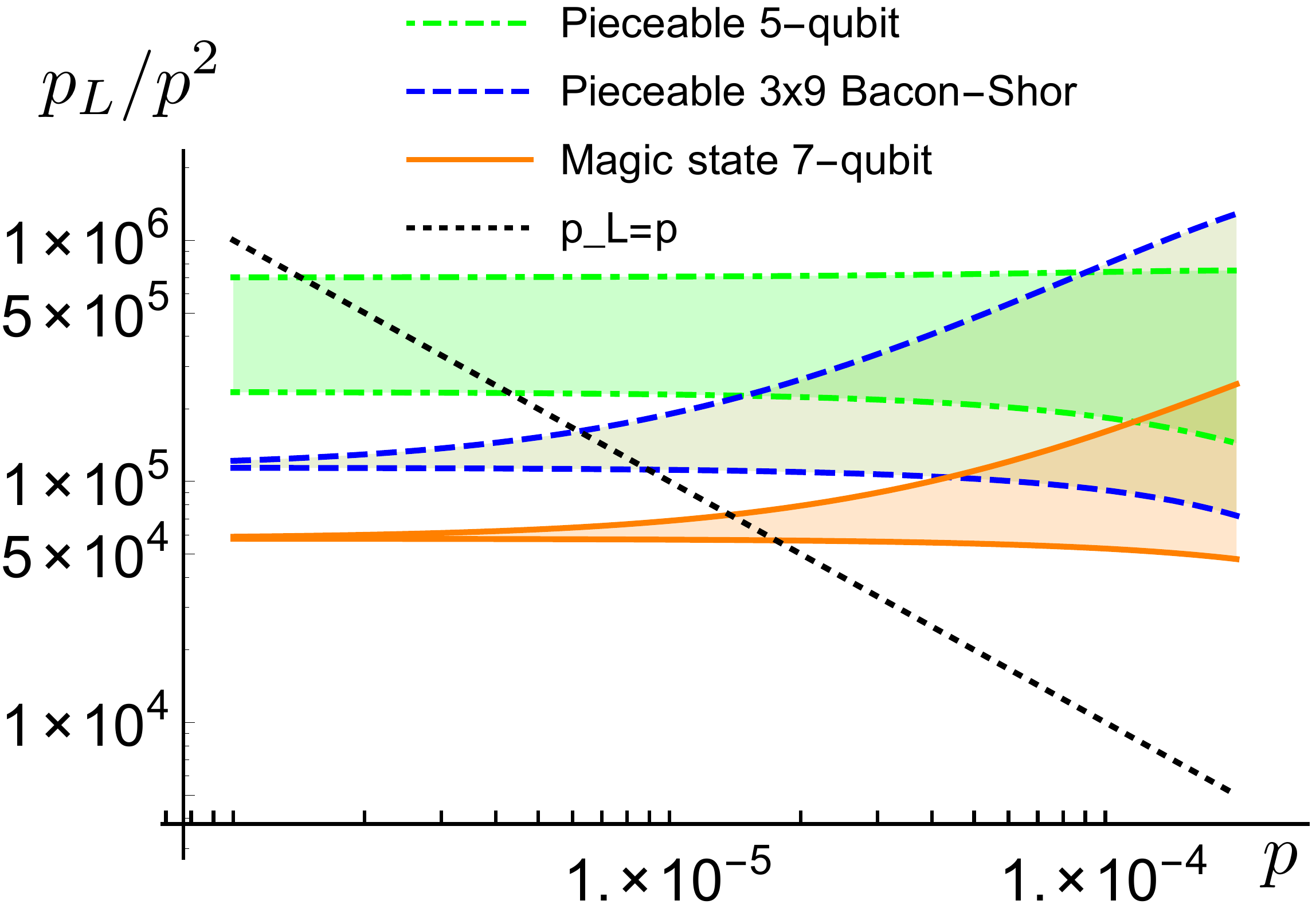}}
\subfigure[]{
\includegraphics[scale=0.25]{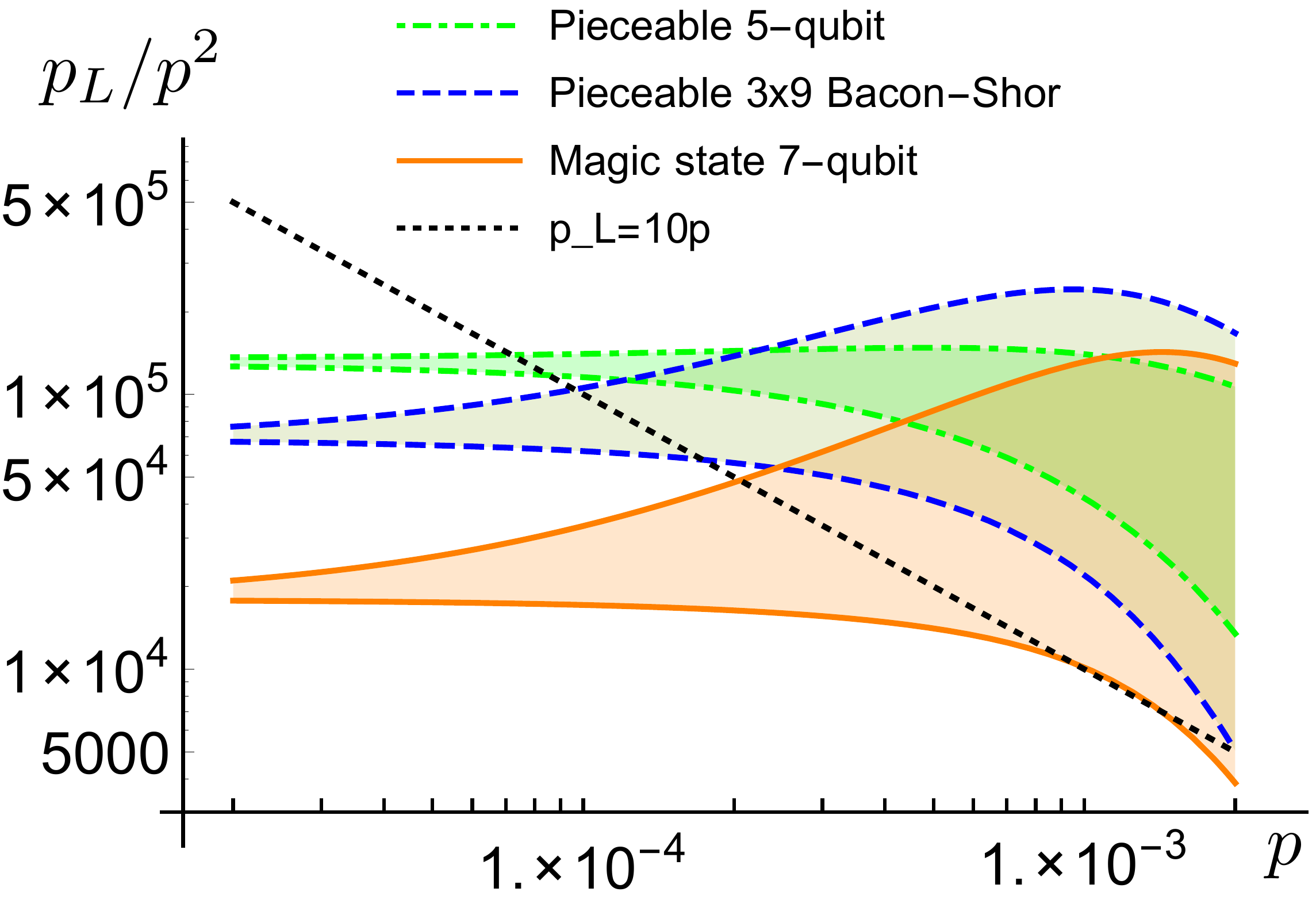}}
\caption{(Color online.) Logical error rates of 3-qubit gate on (a,b) pieceable 7-qubit code (green, dot-dashed), pieceable $3\times 3$ Bacon-Shor code (blue, dashed), (c,d) pieceable 5-qubit code (green, dot-dashed), pieceable $3\times 9$ Bacon-Shor code (blue, dashed), and magic state injection on 7-qubit code (orange, solid) with (a,c) $p_1=p_2=p_3=p$ and (b,d) $10p_1=p_2=0.1p_3=p$ where $p_i$ refers to physical error rate of $i$-qubit gate. Initialization of single-qubit states $\ket{0}$, $\ket{+}$ also fails with probability $p_1$. Dotted line is the ``break-even" line where logical error rate coincides with physical error rate. }
\label{fig:logical_error}
 \end{center}  
\end{figure}


 \begin{table}[htbp]
 \centerline{\begin{tabular}{|c|c|c|c|c|}
    \hline
    & Volume & Qubits & 2-qubit gates & 3-qubit gates\\
    \hline\hline
  Pieceable 5-qubit & 3841 & 364 & 445 & 46\\
  \hline
  Pieceable 7-qubit & 771 & 93 & 162 & 21\\
  \hline
  $3\times 3$ Bacon-Shor & 414 & 81 & 90 & 27\\
  \hline
   $3\times 9$ Bacon-Shor & 1350 & 252 & 306 & 27\\
  \hline
   Magic state & 1352 & 154 & 267 & 14\\
  \hline\hline
  $3\times 3$ BS/Magic & 0.31 & 0.59 & 0.34 & 1.9\\
  \hline
  \end{tabular}}
  \caption{Resource overheads to implement logical CCZ. Volume refers to the circuit volume, which counts all gates weighted by the number of qubits involved. Qubits are the number of physical qubits including data qubits and ancilla qubits where ancilla qubits are assumed to be not reusable. Numbers for the 5-qubit code include all the resources for the adaptive measurements.}
  \label{tbl:resource}
 \end{table}

 \section{Comparison to surface codes} 
 \text{Next,} we compare logical error rate and resource overheads to a local magic-state scheme on surface codes.
We find that the pieceable construction can have a significant advantage in circuit volume in a certain region in terms of physical error rate and target logical error rate.

\subsection{Logical error rates}
Surface codes are known to have high asymptotic threshold, which is 0.1\%-1\% depending on assumptions and error model \cite{Fowler2012a,Fowler2009,Wang2009b,Wang2011,Wootton2012}, and thus they have attracted attention as a candidate for a scalable quantum computer. 
However, having a high asymptotic threshold does not automatically imply that logical error rate is always low for reasonably sized codes.

Firstly, as can be seen in \cite{Fowler2012a}, in the low distance regime the pseudothreshold of the surface code is much smaller than the asymptotic threshold. 
Thus, if the physical error rate is lower than the asymptotic threshold but not below the relevant pseudothreshold, encoding at low distance does not help to reduce the error rate. 

Secondly, the logical error rate of a logical gate can be large even if the error rate for one surface code {\it cycle} is small, because a logical gate is made up of many cycles. Each cycle consists of measuring the complete error syndrome once via measurement qubits, one per stabilizer generator, as in \cite{Fowler2012a}.
Let $\bar{p}_{cycle}$ be the logical error rate for surface code per surface code cycle. 
Let $C_G$ be the number of surface code cycles it takes to implement a logical version of gate $G$.
Then, logical error rate of gate $G$ is $\bar{p}_G\approx C_G \bar{p}_{cycle}$.
Since $C_G \propto d$ and $\bar{p}_{cycle}\propto p^{(d+1)/2}$ where $d$ is the surface code distance, $\bar{p}_{cycle}$ dominates for large distance.
However, when $d$ is small, the contribution to $\bar p_G$ from $C_G$ is not negligible. 
In Appendix~\ref{app:surface_calc}, we find a specific form of $C_G$ for the logical Toffoli gate for two different implementations. 

Fig.~\fig{error_concate_surface} shows logical error rates of a 3-qubit gate on the surface code using a Toffoli state, and upper bounds of logical error rate of pieceable $3\times 3$ Bacon-Shor code and pieceable 7-qubit code in terms of code distance with three different physical error rates.
Upper bounds are obtained by concatenating the function upper bounding the actual rate in Eq.~\eq{bound}. 
Since the 3-qubit gate is the largest component among the components that appear in the logical construction of 3-qubit gate, concatenating the upper bounding error function for the 3-qubit gate upper bounds its error rates at higher concatenation level.
However, because logical 3-qubit gates have an order of magnitude higher error rate than 2-qubit gates and the logical constructions of 3-qubit gates mostly consist of single gates and 2-qubit gates, this upper bound is highly pessimistic. 
A careful analysis taking into account error functions for other types of components and possibly even using better decoding algorithm \cite{Poulin2006,Fern2008a} at a higher levels may greatly reduce estimates of logical error rates.

Nevertheless, in Fig.~\fig{error_concate_surface}, we can see that surface codes have better scaling with distance than pieceable concatenated codes, which should be attributed to the high threshold. However, for small $d$, $C_{\mbox{Toffoli}}$ has a significant contribution, and when $d=3$ the logical error rate of the pieceable constructions is two orders of magnitude lower than that of the surface code.

\begin{figure}[htbp]
\begin{center}
\subfigure[]{
\includegraphics[scale=0.3]{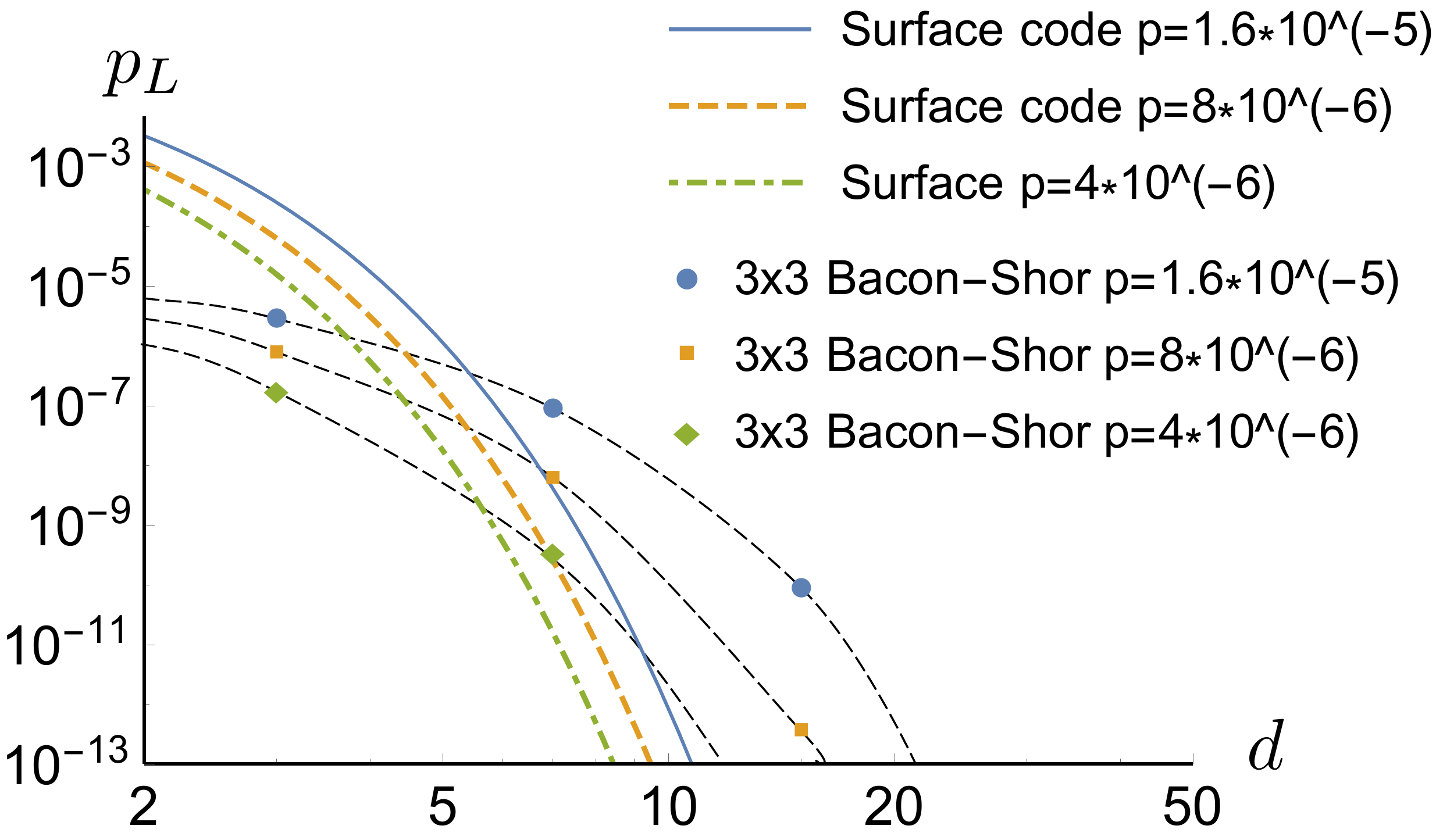}}
\subfigure[]{
\includegraphics[scale=0.3]{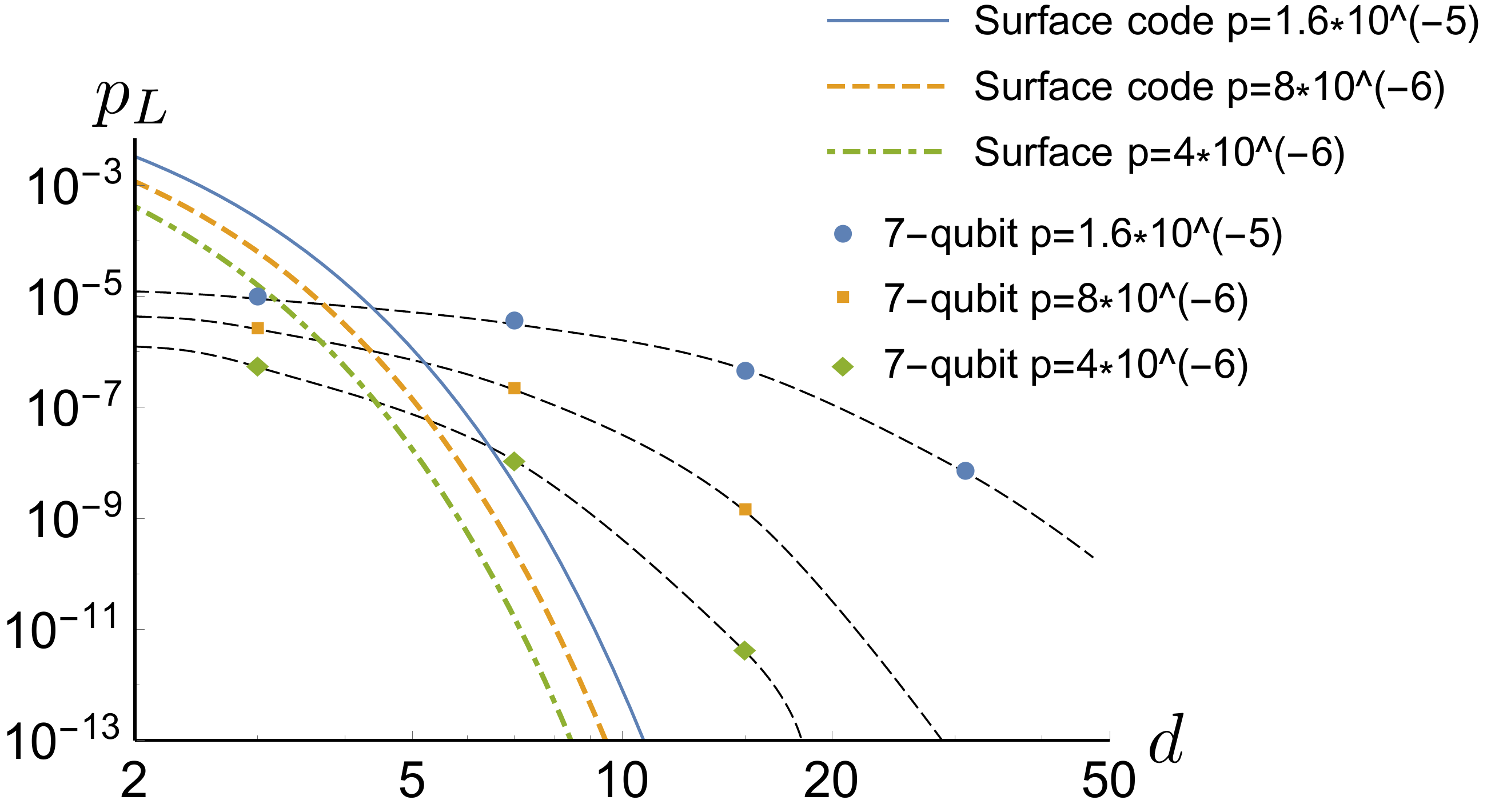}}
\caption{Logical error rates for 3-qubit gate on the surface codes and (a) pieceable $3\times 3$ Bacon-Shor code (b) pieceable 7-qubit code in terms of code distance. Shown rates for pieceable codes are upper bounds obtained by concatenating the upper bounding function from Eq.~\eq{bound}. Black dashed curves are only a guide to the eye.}
\label{fig:error_concate_surface}
\end{center}
\end{figure}

\subsection{Resource overheads}
We also count the circuit volume for implementing logical Toffoli on surface codes.
This allows us to compare the circuit volume between pieceable codes and the surface codes, shown in Fig.~\fig{volume}.
Although surface codes have better scaling with distance, pieceable constructions have a significant advantage until three concatenations. 
This is especially true at distance three, where the difference is three orders of magnitude.

 \begin{figure}[htbp] 
 \begin{center}
\includegraphics[scale=0.3]{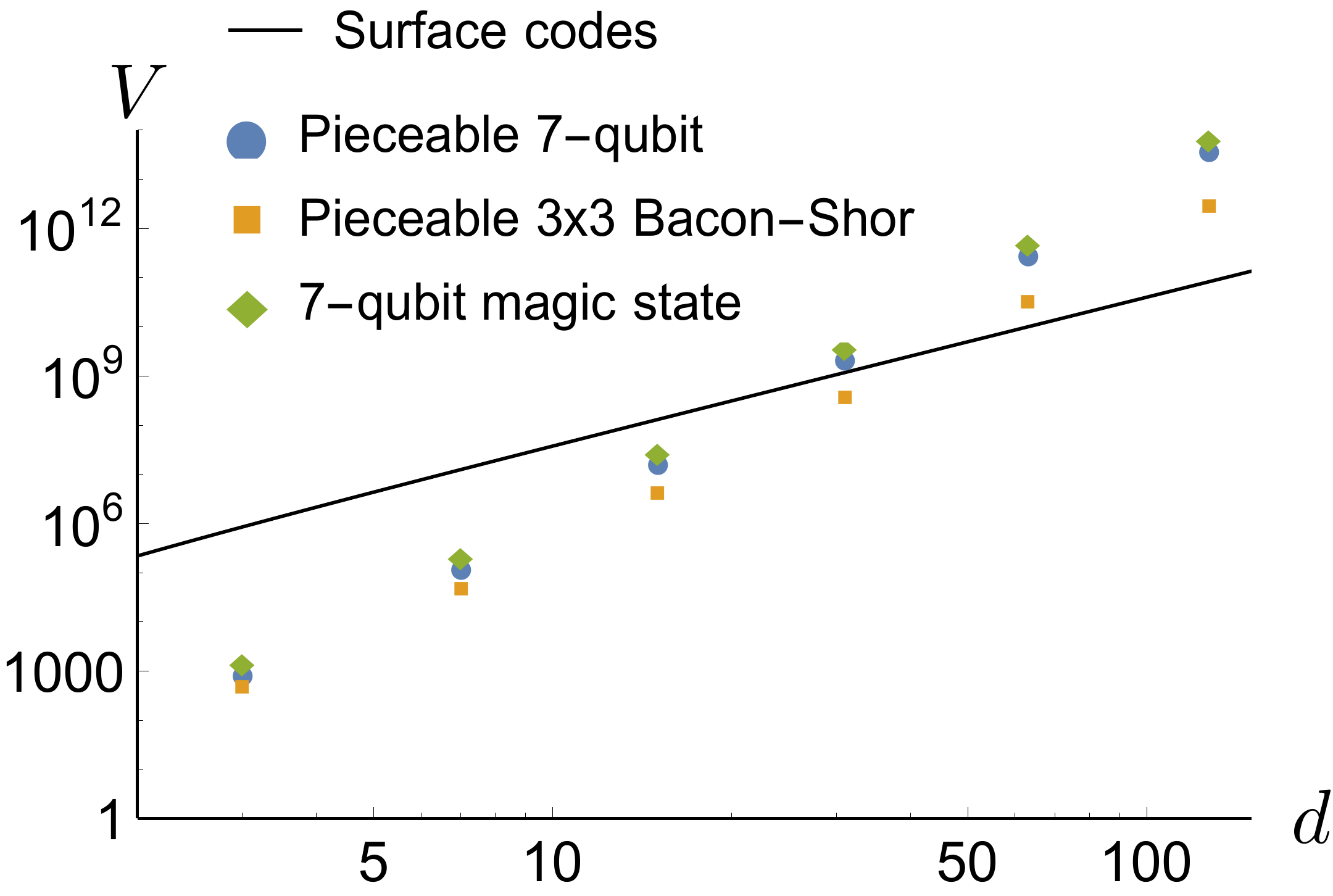}
\caption{Circuit volume for logical 3-qubit gate on pieceable 7-qubit code(circles), pieceable $3\times 3$ Bacon-Shor code(squares), and magic-state scheme on 7-qubit code(diamonds) in terms of code distance. The dots correspond to every concatenation level in the range. Although it may be hard to see the data for pieceable 7-qubit code because they are close to the data for the magic-state scheme, the pieceable 7-qubit has slightly lower volume than the magic-state scheme for every distance shown. 
}

\label{fig:volume}
 \end{center}  
\end{figure}

Consider now the space consisting of pairs (physical error rate, target logical error rate)$\equiv (p,p_T)$. Combining volume and error rate estimates, the region of this space where concatenated pieceable constructions require less circuit volume for implementing Toffoli than the surface code can be obtained. Fig.~\fig{error_volume} shows this region for the pieceable $3\times 3$ Bacon-Shor code and the 7-qubit code.
 It shows that in large range, pieceable $3\times 3$ Bacon-Shor code has advantage in circuit volume over surface code, and the difference can be significant as can be seen in Fig.~\fig{volume}.
This region is actually determined by the upper bound of error rates at the third level concatenation. 
 It is because surface code with distance five has already larger volume than $3\times 3$ Bacon-Shor code with three concatenations as can be seen in Fig.~\fig{volume}.
For the 7-qubit code, Fig.~\fig{volume} shows that a 3-qubit logical gate at two concatenations of the 7-qubit code has less circuit volume than the surface code of any size. Thus, whenever two concatenations are sufficient to achieve the target logical error rate, the 7-qubit code will be advantaged, as is represented by the region in Fig.~\ref{fig:error_volume}. 
Fig.~\fig{volume} also shows that the volume for the 7-qubit code with three concatenations is slightly larger than that for the surface code with distance seven. 
Thus, the surface code is advantaged for the region where distance seven is enough for the surface code but three concatenations are needed for the 7-qubit code, which corresponds to the region between the upper purple region and the lower purple region in Fig.~\ref{fig:error_volume}. 
The 7-qubit code again starts to have advantage over surface code for the region where the surface code needs distance nine whereas the 7-qubit code only needs to be concatenated three times, which corresponds to the lower purple region in Fig.~\ref{fig:error_volume}.


 \begin{figure}[htbp] 
\begin{center}
\includegraphics[scale=0.35]{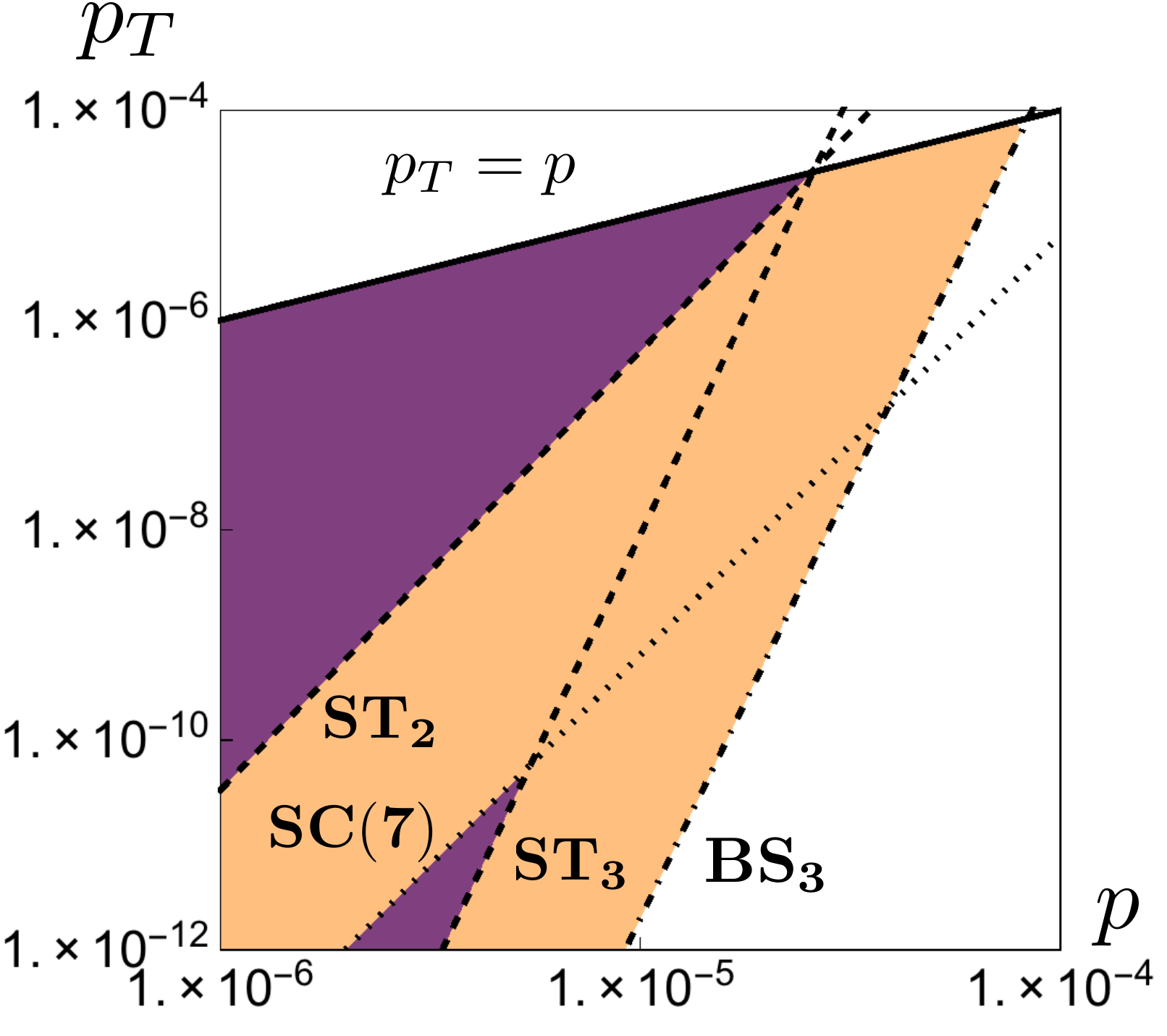}
\caption{(color online.) The region where pieceable $3\times 3$ Bacon-Shor code (orange and purple) and pieceable 7-qubit code (purple) use less volume than surface codes to implement 3-qubit gate to achieve fixed target logical error rate, $p_T$, with fixed physical error rate, $p$. Dashed lines labeled by ${\bf ST_j}$ are upper bounds on logical error rate at the $j^{\text{th}}$-concatenation level of the pieceable 7-qubit code (Steane code). The dotted line labeled by ${\bf SC(7)}$ is the logical error rate for the surface code with distance seven. The dot-dashed line labeled by ${\bf BS_3}$ is an upper bound of logical error rate at the third-concatenation level of the $3\times 3$ Bacon-Shor code. The solid line on the upper boundary is the $p_T=p$ line. 
}
\label{fig:error_volume}
\end{center}  
\end{figure}

 \section{Conclusions} 
 In this paper, we calculated logical error rates and resource overheads of 3-qubit gates using pieceable fault-tolerant constructions, a non-local magic-state scheme (on the 7-qubit code), and a local magic-state scheme (on the surface code).
In comparison with the non-local magic-state scheme, we found that while pieceable constructions have comparable, or even lower logical error rate to the magic-state scheme, the required circuit volume can be as little as 30\%. This suggests that the pieceable construction is a promising complement to schemes relying on magic states.

We also compared the pieceable construction to the surface codes and found that in quite a large region in terms of physical error rates and target logical error rates, pieceable constructions can have significantly lower circuit volume than surface codes.

Although realizing physical components with a small physical error rate such that pieceable constructions have a great advantage is challenging, one should notice that surface codes also have as hard a challenge as this in terms of resource overheads. 
Just as surface codes are good candidates given access to large overheads, the pieceable construction appears to be a good candidate given access to small physical error rates.

Another difference between pieceable constructions and the surface code is locality, i.e.~the constraint that physical gates involved act only between qubits that are neighboring in some chosen low-dimensional layout. 
Although the locality property is desirable in many experimental setups, some systems allow non-local interaction too \cite{Monroe2014}. 
Our result indicates that such a non-local techniques can lead to significant reduction of resource use for quantum error correcting codes.

 \section*{Acknowledgement} 
 R.T. gratefully acknowledges the support of the Takenaka scholarship. T.Y. appreciates the support of the Department of Defense (DoD) through the National Defense Science and Engineering Graduate (NDSEG) Fellowship program.

\appendix

\section{Steane's error correction for arbitrary stabilizer codes}\label{app:generalized_steane}
Here, we describe the error correction used for the leading error correction and trailing error correction of the 5-qubit code. Since the 5-qubit code is not CSS, one might think Steane's error-correction is inappropriate. However, in \cite{Steane1997}, Steane proposes a circuit to do just that for the 5-qubit code. Unfortunately, Steane's construction as written is not quite correct. We present the correct method that works for any stabilizer code. We will also see that this method gives a conceptually simple way to prepare the necessary ancilla state in line with Steane's original proposal \cite{Steane1997}.

Consider a $\llbracket n,k\rrbracket$ stabilizer code $\mathcal{C}$ with stabilizer
\begin{equation}
S=\left(\begin{array}{c|c}S_x&S_z\end{array}\right),
\end{equation}
and logical operators
\begin{equation}
N=\left(\begin{array}{c|c}N_x&N_z\end{array}\right),
\end{equation}
written in symplectic matrix form. That is, ${S_x,S_z\in\mathbb{F}_2^{n-k}\times\mathbb{F}_2^n}$ and ${N_x,N_z\in\mathbb{F}_2^{2k}\times\mathbb{F}_2^{n}}$. Also, if we define ${\Lambda=\left(\begin{smallmatrix}0&I\\I&0\end{smallmatrix}\right)\in\mathbb{F}_2^{2n}\times\mathbb{F}_2^{2n}}$ using $k\times k$ block notation, then the canonical commutation relations are expressed by ${S\Lambda S^T=S\Lambda N^T}$ and ${N\Lambda N^T=A}$ for the ${2k\times 2k}$ matrix $A$ with $1$s on only the antidiagonal.

Following Steane, we propose the circuit in Fig.~\ref{fig:Steane_EC} to extract the syndrome of $\mathcal{C}$. The ancilla state used is twice the size of the code $\mathcal{C}$. The stabilizer of the ancilla state $\ket{\overline{a}}$ can be written
\begin{equation}\label{eq:anc_stabilizer}
S_a=\left(\begin{array}{cc|cc}
S_z&S_x&0&S_z\\
0&0&S_x&S_z\\
0&0&N_x&N_z
\end{array}\right).
\end{equation}
We show that this ancilla state and the circuit in Fig.~\ref{fig:Steane_EC} successfully extract the syndrome without giving information about the logical operators by propagating the observables of the code $\mathcal{C}$ and the stabilizer $S_a$ through the circuit. Begin with,
\begin{align}
\left(\begin{array}{ccc|ccc}aS_x&0&0&S_z&0&0\\bN_x&0&0&N_z&0&0\\0&S_z&S_x&0&0&S_z\\0&0&0&0&S_x&S_z\\0&0&0&0&N_x&N_z\end{array}\right),
\end{align}
where the syndrome is $a\in\mathbb{F}_2^{n-k}$ and logical operator values are $b=\mathbb{F}_2^{2k}$. After the controlled-$Z$ gates,
\begin{equation}
\left(\begin{array}{ccc|ccc}aS_x&0&0&S_z&S_x&0\\bN_x&0&0&N_z&N_x&0\\0&S_z&S_x&S_z&0&S_z\\0&0&0&0&S_x&S_z\\0&0&0&0&N_x&N_z\end{array}\right).
\end{equation}
After the controlled-$X$ gates,
\begin{equation}
\left(\begin{array}{ccc|ccc}aS_x&0&0&S_z&S_x&S_z\\bN_x&0&0&N_z&N_x&N_z\\S_x&S_z&S_x&S_z&0&0\\0&0&0&0&S_x&S_z\\0&0&0&0&N_x&N_z\end{array}\right).
\end{equation}
This is equivalent to the stabilizer,
\begin{equation}
\left(\begin{array}{ccc|ccc}aS_x&0&0&S_z&0&0\\bN_x&0&0&N_z&0&0\\a0&S_z&S_x&0&0&0\\0&0&0&0&S_x&S_z\\0&0&0&0&N_x&N_z\end{array}\right),
\end{equation}
and so we see that measuring all ancilla qubits in the $X$-basis results in a bitstring $m\in\mathbb{F}_2^{2n}$ such that $S\Lambda m=a$.

We note that $\ket{\overline{a}}$ is simply related to a Bell pair $\ket{\Phi}=(\ket{00}+\ket{11})/\sqrt{2}$ encoded in $\mathcal{C}$. If $\text{CX}_{tb}$ denotes $n$ CX gates transversally acting from the top $n$ qubits of the ancilla to the bottom $n$ and $H_t$ denotes $n$ $H$ gates applied to the top $n$ qubits, then $\ket{\overline{a}}=H_t\text{CX}_{tb}\ket{\overline{\Phi}}$. Thus, we can think of $\ket{\overline{a}}$ as an encoded Bell pair that has been ``transversally disentangled". Circuit identities can be used to rearrange Fig.~\ref{fig:Steane_EC} to Knill's error-correction \cite{Knill2005a}. Also, if $\mathcal{C}$ is CSS, Fig.~\ref{fig:Steane_EC} reduces to Steane's error-correction for CSS codes \cite{Steane1998a}.

Steane's original proposal for non-CSS error-correction \cite{Steane1997} omitted the $S_z$ on the right side of the first row of Eq.~\eqref{eq:anc_stabilizer}. Doing the same calculation as above shows that this will not succeed in measuring the syndrome. Steane's proposal suggested that the ancilla state would always be CSS for any code. This, unfortunately, is not true. Indeed, the 7-qubit code from \cite{Yoder2017} has an ancilla that is not even local-Clifford (LC) equivalent to a CSS state.

However, there are non-CSS codes for which $S_a$ is LC equivalent to a CSS code. The 5-qubit code with stabilizer
\begin{equation}
S_5=\left(\begin{array}{ccccc|ccccc}
1&0&1&0&0&0&0&0&1&1\\
0&1&0&0&1&0&0&1&1&0\\
1&0&0&1&0&0&1&1&0&0\\
0&0&1&0&1&1&1&0&0&0
\end{array}\right)
\end{equation}
is one of these. Indeed, $S_a$ can be written using only $Y$-type and $Z$-type generators. This allows us to prepare the ancilla using Fig.~\ref{fig:Steane_prep}, and verify the ancilla against single circuit faults using Fig.~\ref{fig:Steane_verification}, which are both standard constructions for CSS states \cite{Steane1998a,Cross2009}.

\begin{figure}[htbp] 
\begin{center}
\includegraphics[scale=0.3]{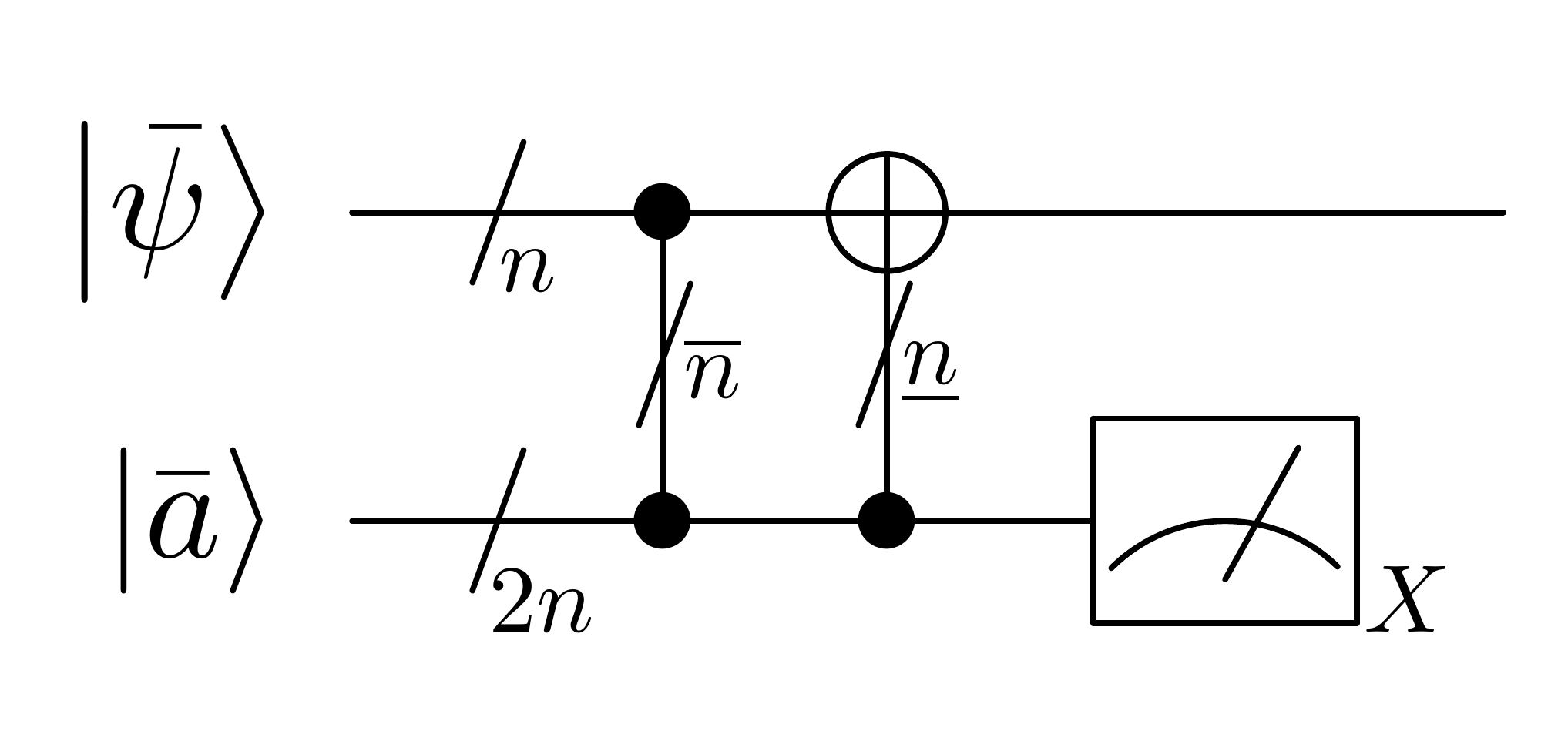}
\caption{Circuit for Steane's error correction on a non-CSS code. $\ket{\bar{\psi}}$ is an arbitrary encoded state, and $\ket{\bar{a}}$ is the $2n$-qubit ancilla state from Eq.~\eqref{eq:anc_stabilizer}. The notation $\overline{n}$ and $\underline{n}$ means that the $n$ CZ gates are transversally coupled to the top $n$ qubits in the ancilla, and that the $n$ CNOT gates are transversally coupled to the bottom five qubits in the ancilla. Measurement is done transversally, from which the syndrome can be classically computed.}
\label{fig:Steane_EC}
\end{center}  
\end{figure}

 \begin{figure}[htbp] 
\begin{center}
\includegraphics[scale=0.4]{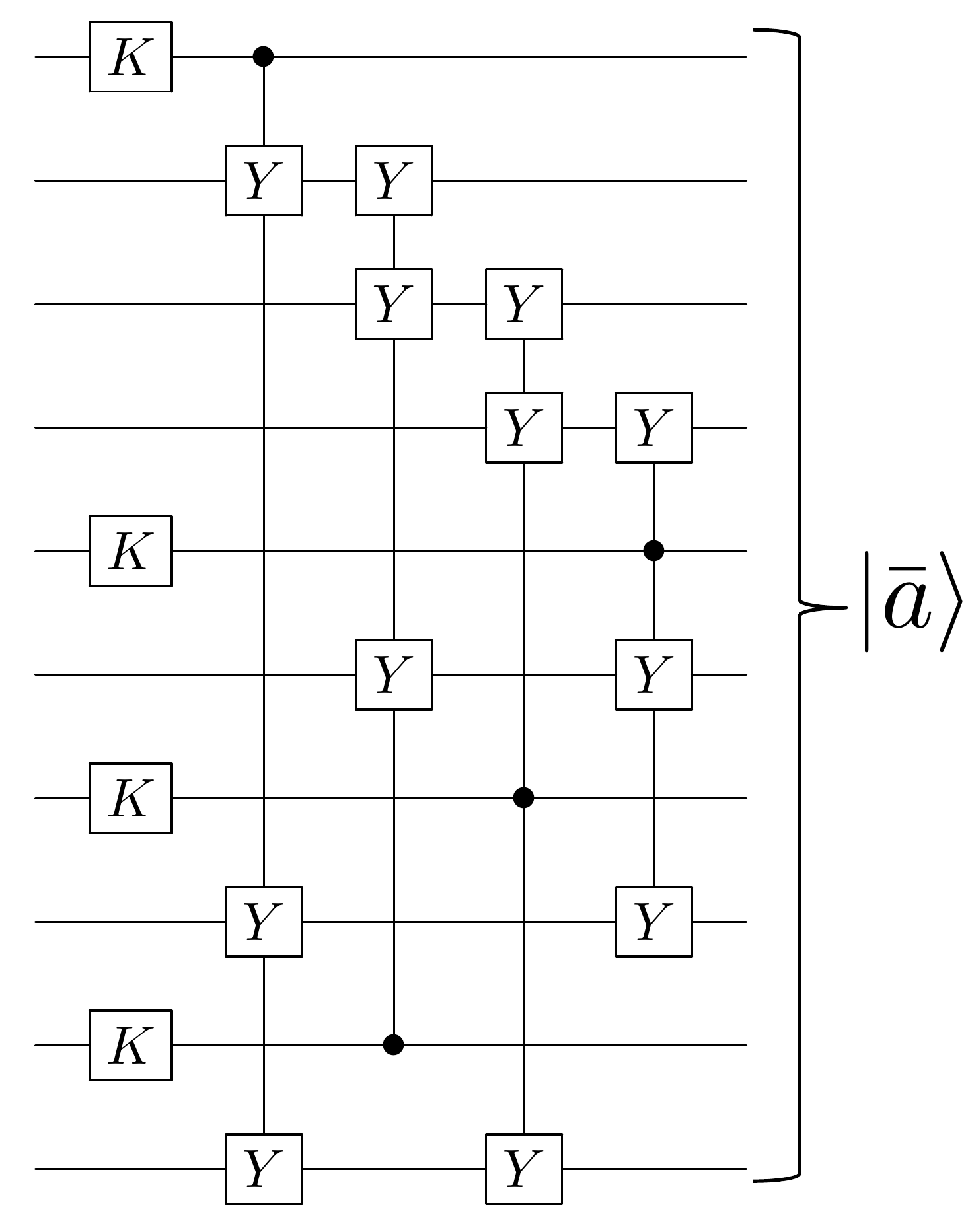}
\caption{Preparing the error-correction ancilla state for the 5-qubit code for use in Fig.\fig{Steane_EC}. Input states are all prepared in $\ket{0}$.}
\label{fig:Steane_prep}
\end{center}  
\end{figure}

 \begin{figure}[htbp] 
\begin{center}
\includegraphics[scale=0.3]{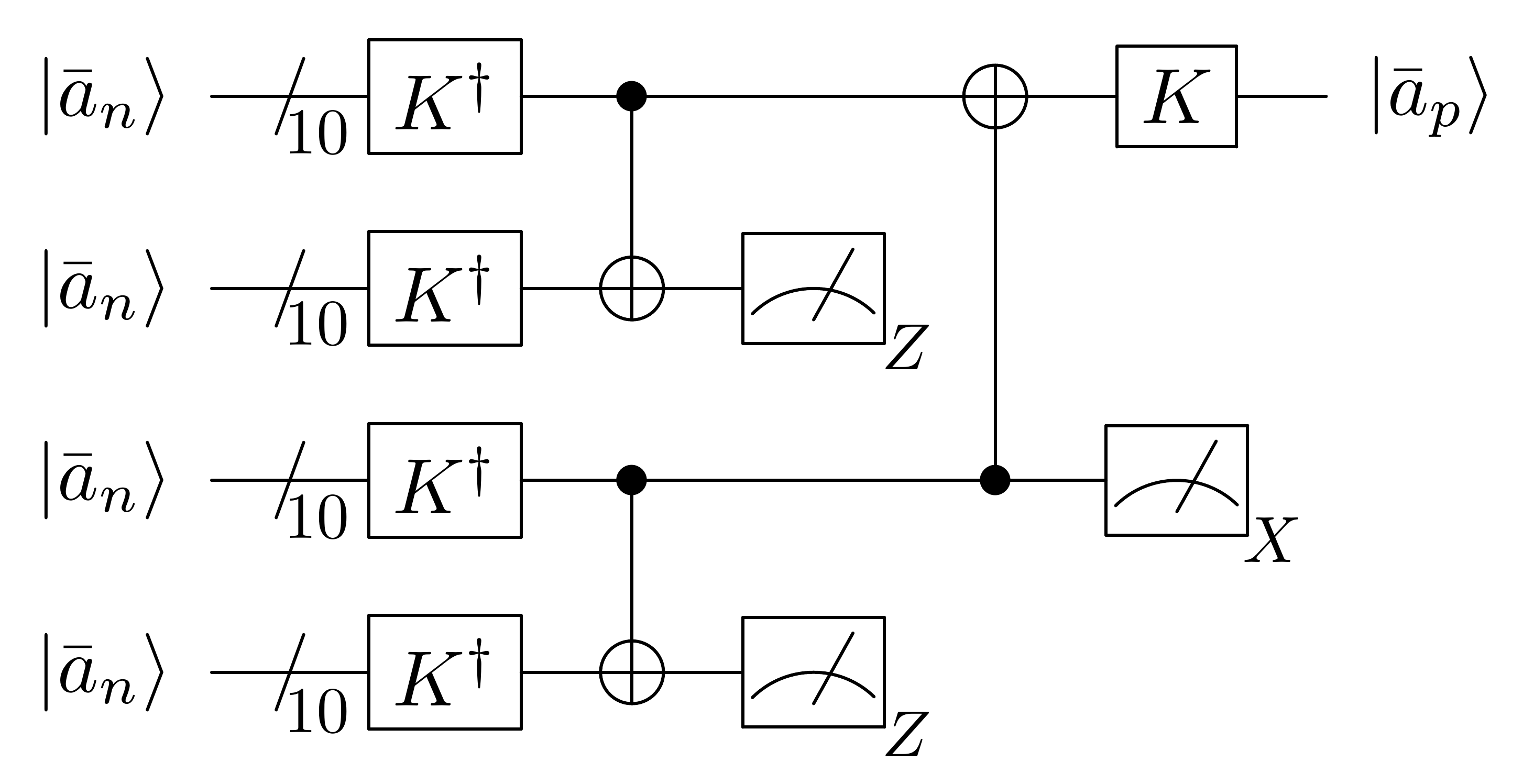}
\caption{Verification circuit for the ancilla state prepared by the circuit in Fig.\fig{Steane_prep}. $\ket{\bar{a}_n}$ is a noisy ancilla state that needs to be verified, and $\ket{\bar{a}_p}$ is a purified one.}
\label{fig:Steane_verification}
\end{center}  
\end{figure}

\section{Details of the simulation for logical error estimation}\label{app:sim_details}
Here, we describe some techniques used in the estimation of logical error rates. 

For reasons of simulation efficiency, only errors originating from at most two faults are considered, but all such errors are counted. For Clifford circuits, propagating the Pauli errors resulting from circuit depolarizing noise can be done simply using the Gottesman-Knill theorem \cite{Gottesman1998a}. However, some of our circuits are built from non-Clifford CCZ gates. In this case, a tracked error is modified to include controlled-$Z$ (CZ) terms. A Pauli error that propagates through $m$ CCZ gates picks up at most $m$ CZ terms (some may cancel). Upon measurement (e.g.~in the error-correction circuits), the CZ terms must be broken down into a sum of Paulis, only some of which flip measurement bits to cause a signal. 
We treat each term as different error element with the probability equal to the square of the amplitude of the term. 

There is a subtlety in breaking down CZ errors. 
As a sum of Pauli terms, a CZ error is written $(II+ZI+IZ-ZZ)/2$.
If there are multiple CZ errors, this Pauli sum has every possible combination of $I$ and $Z$ on the qubits on which CZ errors are applied, each with a plus or minus sign.
Thus, $m$ CZ errors applied on different qubit pairs are decomposed into a Pauli sum with $4^m$ terms. 
If we treated each term as different error element at this point, each term would be assigned the probability square of the amplitude. 
However, some terms may be equivalent to other terms up to stabilizers.
Such terms should interfere coherently. 
In the simulation of pieceable CCZ on $3\times 3$ Bacon-Shor code, all the terms in the Pauli sum are rewritten in an unambiguous way up to stabilizers, and terms interfere before assigning them a probability.

Now that we recognize the subtle issue as the coherent addition of the Pauli terms, we argue that it does not affect the logical error rate except of the $3\times 3$ Bacon-Shor code case.
Firstly, note that the coherent addition can only happen when the number of qubits in one block on which CZ errors are applied is more than or equal to the weight of stabilizers.
This is because if stabilizers have higher weight, multiplying a stabilizer necessarily gives extra Paulis on the qubits that are not affected by CZ errors.
It prevents the term multiplied by a stabilizer being the same as another term in the Pauli sum.
In pieceable CCZ circuit on the 5-qubit code, CZ errors only occur on the three qubits, the support of logical $Z$. Since stabilizers are weight four, the coherent addition will not happen for the above reason.
For the pieceable CCZ circuit on the 7-qubit code, we argue that although the coherent addition may happen, it will not affect logical error rate. 
Since stabilizers for the 7-qubit code are weight four, the coherent addition could happen only when two $X$ errors go through in the same block in the first piece. 
However, these $X$ errors cannot be corrected because the 7-qubit code is a perfect CSS code.
Thus, all the error elements where the coherent addition could happen end up with logical errors regardless, and it does not matter whether we accurately interfere the terms. 
For the pieceable CCZ circuit on the $3\times 9$ Bacon-Shor code, the situation is similar to the 7-qubit code case; the coherent addition could happen, but will not affect the logical error rate. 
Since CCZs on the $3\times 9$ Bacon-Shor code are transversal, the number of qubits in the same code block on which CZ errors are applied is at most two. 
Since weight-two $Z$-gauge operators are aligned along a row, the coherent addition could only happen when two CZ errors are applied on the two qubits in the same row in some code block. 
However, all the terms in the Pauli sum for the CZ errors in that block are $Z$-type errors whose weight is less than or equal to two, and whose support is in the same row.
Since weight-one errors can be corrected by the standard error correction, and weight-two errors in the same row are equivalent to the identities up to stabilizers for the $Z$-gauge Bacon-Shor code, the terms in the Pauli sum are all correctable when the concerned coherent addition could happen.  
Thus, it will not affect the logical error rate. 
 
Considering CZ errors as a Pauli sum is inefficient -- $m$ CZ terms lead to $4^m$ Pauli addends.
However, in the simulation, we do not actually break down all the CZ errors. 
Under certain cases, we definitely know that the final error correction will succeed to correct the CZ error.
One of such cases is that the CZ error is applied over different code blocks and those code blocks do not have any $Z$ errors.
The other case is that the CZ error is applied in one code block, there are no $Z$ errors in the block, and an intermediate error correction notifies the correct locations that the CZ error is applied over.

Also, we can reduce the number of CZ errors by removing harmless CZ errors before the measurements in the final error correction take place. 
A harmless error is one that does not affect encoded states.
When errors are only Paulis, like in the circuits that only consist of Clifford gates, such errors are just stabilizers.
The following theorem generalizes the condition for the harmless errors to non-Pauli case.
\begin{thm}\label{thm:harmless}
 Let $E$ be an error operator, $S=\left<g_1,\dots ,g_{n-k}\right>$ be the stabilizers, $\left<g_{n-k+1},\dots,g_{n+k}\right>$ be the logical operators of the code, and $\ket{\bar{\psi}}$ be an encoded state. If $g_i^{\dagger}E^{\dagger}g_iE\in S$ for all $i=1,\dots,n+k$, then $E\ket{\bar{\psi}}=\ket{\bar{\psi}}$ up to global phase. 
\end{thm}
\begin{proof}
 By the assumption, there exists a stabilizer $s_l$ such that $g_iE=Eg_is_l, \forall i$.
For $i=1,\dots,n-k$, since
\begin{eqnarray}
g_iE\ket{\bar{\psi}}=Eg_is_l\ket{\bar{\psi}}=E\ket{\bar{\psi}},
\end{eqnarray}
$E$ preserves the codeword space. Now for $i=n-k+1\dots n+k$, let $\ket{g_i^{(\pm)}}$ be the eigenstate of the logical operator $g_i$ with eigenvalue $\pm 1$, then
\begin{eqnarray}
g_i E\ket{g_i^{(\pm)}}=E g_i s_l\ket{g_i^{(\pm)}}=\pm E\ket{g_i^{(\pm)}}.
\end{eqnarray}
Thus, $E$ also preserves the logical space. 
\end{proof}
\noindent This theorem allows us to ignore the CZ errors that satisfy the above condition, which greatly reduces the computational task.

When intermediate error corrections are present, CZ errors need to be broken down according to the Pauli sum in the intermediate error corrections, and need to be propagated until the error correction at the end. 
If the number of intermediate corrections is zero or one, it is rather easy to deal with, because the number of error elements due to the CZ errors that need to be propagated until the end is limited.
Actually, except the pieceable 5-qubit code, all the CZ errors that do not satisfy the condition of Theorem~\ref{thm:harmless} were broken down upon measurement and tracked to see if they end with a logical error. 

For the 5-qubit code, to reduce the computational demand, we take the rule where we declare an error to be a logical error as soon as some CZ errors are measured in an intermediate error correction. 
Although this strategy would cause some overestimation of the logical error rate, we argue that the probability that CZ errors are measured in an intermediate error correction is rather small. 
CZ errors are measured in an intermediate error correction in the following two cases.
The first case is that an $X$ or $Y$-type error is caught by a CCZ gate in the adaptive nonconstant-stabilizer measurement. 
It is described in \cite{Yoder2016c} that the adaptive nonconstant-stabilizer measurement is only triggered when some constant stabilizer measurements click due to an $X$ or $Y$-type error only for a single code block.  
The adaptive measurement may contain CCZ gates connected between the ancilla block and the code blocks whose constant stabilizers did {\it not} click.   
Thus, an $X$ or $Y$-type error is caught by a CCZ gate in the adaptive measurement only when an $X$ or $Y$ error triggers the adaptive measurement, the constant measurement in different code blocks fail with $X$ or $Y$ type error, and it goes to a CCZ gate in the adaptive measurement.
The second case is that a CZ error is caught by a CNOT gate in the adaptive measurement.
Note that CZ error only happens when an $X$ or $Y$-type error propagates through the CCZ gates in the code blocks. CZ errors are then present in code blocks other than the one in which the $X$ or $Y$-type error exists. Also, CNOT gates in the adaptive measurement could be only applied to the code block whose constant stabilizers click. 
Thus, a CZ error is caught by a CNOT gate in the adaptive measurement only when an $X$ or $Y$-type error generates CZ errors in different code blocks, a later CCZ gate fails to cancel the first $X$ or $Y$-type error and generate another $X$ or $Y$-type error in the other code block that will make the constant measurement click, and the CZ error goes into CNOT gate in the adaptive measurement. 
These two cases are realized in very restricted situations, so the contribution to the total logical error rate from these cases would be rather small. 

Another situation arises with two or more intermediate corrections.
The pieceable construction on the 5-qubit code has multiple intermediate corrections, and they detect $X$ errors and notify possible error locations to the final error correction so that the final error correction can correct up to weight two located errors. 
However, multiple faults can cause two intermediate error corrections to incorrectly notify more than two locations to the final error correction.
We declare those elements to be logical errors.

\section{Details on the error polynomials}\label{app:error_polys}
Here, we describe how to obtain Eq.~\eq{fail_acc}-\eq{rej} from the exact counting.
We first consider Eq.~\eq{fail_acc}, the probability that one or two faults occur and that pattern is accepted by all the verification modules through the propagation, but ends up with a logical error.
Due to the fault-tolerant property, a single fault never causes a logical error. Thus, it suffices to consider the cases when two faults occur. 
In the simulation, each combination of two-fault patterns is assigned a probability $\left(\frac{p_{r}}{4^{r}-1}\right)\left(\frac{p_{s}}{4^{s}-1}\right)$ if the faulty components are an $r$-qubit gate and an $s$-qubit gate.
We propagate all the errors until the end and sum up the probabilities of the errors that lead to logical errors. 
During the propagation, these errors may encounter verification processes. 
If they are accepted by the verification, we keep propagating them. Otherwise, we stop propagating them so that they do not contribute to the logical error rate. 
Let $Q_{\rm fail, acc}$ denote the estimated logical error rate. Since each physical error rate is either $p_1, p_2$ or $p_3$, it looks like
\begin{equation}
Q_{\rm fail,acc}=\sum_{r=1}^3\sum_{s\geq r}^3F_{rs}^{(2)} p_r p_s.
\end{equation}
Let $n_r$ be the total number of $r$-qubit gates. 
Since we assume that different components fail independently, Eq.~\eq{fail_acc} is obtained as
\begin{equation}
 P_{\rm fail,acc}^{(2)}=\left[\Pi_{t=1}^3 (1-p_t)^{n_t}\right]\left(\sum_{r=1}^3\sum_{s\geq r}^3F_{rs}^{(2)} \frac{p_r}{1-p_r} \frac{p_s}{1-p_s} \right).
 \label{eq:fail_acc_formula}
\end{equation}

Similarly, $Q_{\rm succ, acc}$, the sum of the assigned probability of the patterns that are accepted by all the verification modules and do not cause a logical error, looks like
\begin{equation}
 Q_{\rm succ,acc}=\sum_{r=1}^3 S_{r}^{(1)}p_r+\sum_{r=1}^3\sum_{s\geq r}^3S_{rs}^{(2)} p_r p_s
\end{equation}
and Eq.~\eq{succ_acc} is obtained as
\begin{eqnarray}
 &&P_{\rm succ,acc}^{(2)}=\left[\Pi_{t=1}^3 (1-p_t)^{n_t}\right]\cdot \nonumber\\
&&\left(1+\sum_{r=1}^3 \frac{S_{r}^{(1)} p_r}{1-p_r}+\sum_{r=1}^3\sum_{s\geq r}^3 \frac{S_{rs}^{(2)} p_r p_s}{(1-p_r)(1-p_s)} \right)
\label{eq:succ_acc_formula}
\end{eqnarray}
The patterns that are not counted in either $Q_{\rm fail,acc}$ or $Q_{\rm succ,acc}$ are rejected in some verification module. Thus, we obtain Eq.~\eq{rej} as
\begin{eqnarray}
 &&P_{\rm rej}^{(2)}=\left[\Pi_{t=1}^3 (1-p_t)^{n_t}\right]\cdot \nonumber\\
&&\left(\sum_{r=1}^3 \frac{A_r^{(1)} p_r}{1-p_r}+\sum_{r=1}^3\sum_{s> r}^3 \frac{A^{(2)}_{rs} p_r p_s}{(1-p_r)(1-p_s)}\right)
\end{eqnarray}
where
\begin{eqnarray}
A^{(1)}_r&=&n_r-S_{r}^{(1)}\\
A^{(2)}_{rs}&=&
\begin{cases}
n_r n_s - F_{rs}^{(2)}-S_{rs}^{(2)} & (r\neq s)\\
\binom{n_r}{2} - F_{rr}^{(2)}-S_{rr}^{(2)} &(r=s)
\end{cases}
\end{eqnarray}

Special care is required for 5-qubit code because $n_r$ cannot be definitely determined because of the adaptive measurements.
Note that at most two adaptive measurements are triggered when one or two faults occur.
Thus, taking $n_r$ that includes two largest adaptive measurements, which are the ones with 13-CAT and 9-CAT, the lower bound in Eq.~\eq{bound} still holds.
Instead of Eq.~\eq{succ_acc_formula}, we take
\begin{eqnarray}
 &&P_{\rm succ,acc}^{(2)}=\left[\Pi_{t=1}^3 (1-p_t)^{n_t'}\right]\cdot \left(1+\sum_{r=1}^3 \frac{S_{r}^{(1)} p_r}{1-p_r}\right)\nonumber \\
 &&+\left[\Pi_{t=1}^3 (1-p_t)^{n_t}\right]\cdot\left(\sum_{r=1}^3\sum_{s\geq r}^3 \frac{S_{rs}^{(2)} p_r p_s}{(1-p_r)(1-p_s)} \right)
\end{eqnarray}
where $n_r'$ is the number of fault locations not including adaptive measurements.

For the 5-qubit code we also use
\begin{equation}
 P_{\rm acc}=\Pi_j P_{{\rm acc},j}=\Pi_j (1-P_{{\rm rej},j})<\Pi_{j'} (1-P_{{\rm rej},j'}^{(2)})
\end{equation}
where $j$ is taken over all the verification modules and $j'$ is taken over all the verification modules except adaptive measurements. 

The following are the obtained values for the parameters for each construction.
 
\begin{itemize}
\item $3\times 3$ Bacon-Shor
\begin{equation}
n_1=252,n_2=180,n_3=27
\end{equation}
\begin{equation}
S_{r}^{(1)}=(252,180,27)
\end{equation}
\begin{eqnarray}
F_{rs}^{(2)}&=&
 \begin{pmatrix}
  4216.8& 4271.9 & 783.5  \\
  & 1194.5 & 461.5 \\
  &  & 34.9
 \end{pmatrix}\\
S_{rs}^{(2)}&=&
 \begin{pmatrix}
  27409.2& 41088.1  & 6020.5\\
  & 14915.5 & 4398.5 \\
  &  & 316.1 
 \end{pmatrix}
\end{eqnarray}

\item Pieceable 7-qubit
\begin{equation}
n_1=648,n_2=480,n_3=21
\end{equation}
\begin{equation}
S_{r}^{(1)}=(383,224,21)
\end{equation}
\begin{eqnarray}
F_{rs}^{(2)}&=&
 \begin{pmatrix}
  13258.4& 12722.6 & 3581.4  \\
  & 3077.3 & 1855.3 \\
  &  & 176.7
 \end{pmatrix}\\
S_{rs}^{(2)}&=&
 \begin{pmatrix}
  56460.9& 68953.7  & 4461.6\\
  & 20748.8 & 2848.7 \\
  &  & 33.3 
 \end{pmatrix}
\end{eqnarray}

\item $3\times 9$ Bacon-Shor
\begin{equation}
n_1=2736,n_2=864,n_3=27
\end{equation}
\begin{equation}
S_{r}^{(1)}=(1524,566.4,27)
\end{equation}
\begin{eqnarray}
F_{rs}^{(2)}&=&
 \begin{pmatrix}
  52074 & 43098.4 & 7049.2  \\
  & 8663.0 & 2968.1 \\
  &  & 183.3
 \end{pmatrix}\\
S_{rs}^{(2)}&=&
 \begin{pmatrix}
  1013640 & 748636  & 34098.8\\
  & 138296 & 12324.7 \\
  &  & 167.7 
 \end{pmatrix}
\end{eqnarray}

\item Pieceable 5-qubit 
\begin{eqnarray}
n_1&=&3365,n_2=1228,n_3=41\\
n_1'&=&2967, n_2'=1152, n_3'=27
\end{eqnarray}
\begin{equation}
S_{r}^{(1)}=(1475,457.6.,27)
\end{equation}
\begin{eqnarray}
F_{rs}^{(2)}&=&
 \begin{pmatrix}
  113030.0 & 85261.6 & 14679.2  \\
  & 16067.4 & 5551.4 \\
  &  & 332.5
 \end{pmatrix}\\
S_{rs}^{(2)}&=&
 \begin{pmatrix}
  639301.0 & 482043.0  & 20392.4\\
  & 90716.0 & 7554.7 \\
  &  & 59.3 
 \end{pmatrix}
\end{eqnarray}

\item 7-qubit with magic state
\begin{equation}
n_1=1138,n_2=743,n_3=14
\end{equation}
\begin{equation}
S_{r}^{(1)}=(612,324.3,6.9)
\end{equation}
\begin{eqnarray}
F_{rs}^{(2)}&=&
 \begin{pmatrix}
  25436.5 & 24565.9 & 1078.5  \\
  & 6232.4 & 521.1 \\
  &  & 26.9
 \end{pmatrix}\\
S_{rs}^{(2)}&=&
 \begin{pmatrix}
  154650 & 166625  & 3921.4\\
  & 44308.3 & 2178.5 \\
  &  & 18.6 
 \end{pmatrix}
\end{eqnarray}
\end{itemize}
 
\section{Transformation matrix for volume calculation}\label{app:volume_calc}
As explained in the main text, the circuit volume for concatenated codes at higher concatenation level is described by a transformation matrix $A$ where $A_{ij}=N^{G_j}_{G_i}$. 
We show the matrices for pieceable $3\times 3$ Bacon-Shor code, pieceable 7-qubit code, and 7-qubit with magic state, which are denoted by $A_{pBS}$,$A_{p7}$, $A_{m7}$ respectively.
We take the following order for gates; {\bf G}=\{3-qubit gate, 2-qubit gate, single qubit gate, $\ket{0}$ and $\ket{+}$ preparation, $X$ basis and $Z$ basis measurement\}.
For preparation of $\ket{\bar{0}}$ and $\ket{\bar{+}}$ on 7-qubit code, we use the method proposed by Goto \cite{Goto2016}, which requires just one additional ancilla. 
\begin{equation*}
A_{pBS}=
 \begin{pmatrix}
 27 & 90 & 45 & 54 & 54 \\
 0 & 69 & 30 & 36 & 36 \\
 0 & 30 & 24 & 18 & 18 \\
 0 & 6 & 3 & 9 & 0 \\
 0 & 0 & 0 & 0 & 9 
 \end{pmatrix}
 \label{eq:pBS_trans}
\end{equation*}
\begin{equation*}
A_{p7}=
 \begin{pmatrix}
 21 & 162 & 240 & 72 & 72 \\
 0 & 79 & 104 & 32 & 32 \\
 0 & 36 & 59 & 16 & 16 \\
 0 & 11 & 22 & 8 & 1 \\
 0 & 0 & 0 & 0 & 7 
 \end{pmatrix}
 \label{eq:p7_trans}
\end{equation*}
\begin{equation*}
A_{m7}=
 \begin{pmatrix}
 14 & 267 & 504 & 136 & 136 \\
 0 & 79 & 104 & 32 & 32 \\
 0 & 36 & 59 & 16 & 16 \\
 0 & 11 & 22 & 8 & 1 \\
 0 & 0 & 0 & 0 & 7 
 \end{pmatrix}
 \label{eq:m7_trans}
\end{equation*}

\section{Detailed resource analysis for surface code}\label{app:surface_calc}
We describe the detailed resource analysis to implement logical Toffoli gate on the surface code.
There are mainly two ways to do it, synthesizing a Toffoli gate using Clifford gates and $T$ gates, and injecting a logical Toffoli state by gate teleportation.

Consider the first method, in the context of the Toffoli implementation proposed by Jones \cite{Jones2013} using four $T$ gates. The
$T$ gates are implemented by $\ket{T}$ state and gate teleportation where $\ket{T}$ state is purified by a distillation protocol. 
We use the 15-1 protocol~\cite{Bravyi2005a,Fowler2012a} which reduces error rates of $\ket{T}$ from $\oo(p)$ to $\oo(p^3)$, because it requires the smallest circuit volume compared to others \cite{Bravyi2012a,Meier2013,Reichardt2004a}.
Since the region of the physical error rate that pieceable construction helps to reduce error rate is $p<10^{-4}$ as can be seen in Fig.\fig{logical_error}, the logical error rate of the magic state distilled once is $<10^{-12}$.
Although the reduction in error rate may not be sufficiently low depending on the goal logical error rate, one distillation already gives large overheads. 
Thus, we consider the circuit volume for one distillation as a lower bound and proceed the discussion.
 
It may come as a surprise that other distillation protocols with better conversion rate between noisy magic state and purified magic state have larger circuit volume. 
It comes from that Hadamard gate and phase gate are not transversal on the surface code.
For implementing the Hadamard gate or phase gate fault-tolerantly, some non-trivial techniques, such as state injection, lattice surgery \cite{Horsman2012a}, code deformation \cite{Bombin2009a}, or surface folding \cite{Moussa2016}, are required. 
These take many surface code steps, which affect the circuit volume.
Even though conversion rate between noisy $T$ state and purified $T$ state is high, if it requires many costly Clifford gates, the circuit volume will be large. 
Especially in the case when only one distillation is required, a poor conversion rate does not hurt circuit volume that much.  

Let us analyze the number of surface code cycles and circuit volume for each gate that are necessary to implement the logical Toffoli gate. 
Let $C_G$ and $V_G$ be surface code cycles and circuit volume it takes to implement $G$.
We discuss circuit volume in units of [qubit$\cdot$cycle] and then convert it to [qubit$\cdot$step] using the fact that one surface code cycle consists of six steps \cite{Fowler2012a}.
Also, let $d$ be surface code distance, and $n=(2d-1)^2$ be the number of physical qubits on a surface.
Necessary components here are \{$\ket{\bar{0}}$ and $\ket{\bar{+}}$ preparation, CNOT, Hadamard, Phase\}.

For logical state preparation, we initialize a surface with physical $\ket{0}$ for $\ket{\bar{0}}$ preparation, and $\ket{+}$ for $\ket{\bar{+}}$ preparation.
After $d$ rounds of error correction, an appropriate recovery can be determined to prepare the desired logical state fault-tolerantly. 
Thus, we find $C_{prep}=d$, $V_{prep}=nd$. 

The CNOT gate can be transversally implemented if we allow non-locality or a 3D layered architecture. 
However, since one of the striking features of surface codes is local interactions in a 2D architecture, we use lattice surgery to implement the CNOT gate \cite{Horsman2012a}. 
First, prepare a surface with $\ket{\bar{+}}$ state between the control surface and the target surfaces. 
The control surface and the intermediate surface are merged while obtaining measurement syndromes.
This corresponds to $\bar{Z}\bar{Z}$ measurement.  
After that, the surface is split into two original surfaces and the intermediate surface is merged to target surface, which corresponds to $\bar{X}\bar{X}$ measurement. 
It ends with splitting it into the two original surfaces. 
Since merger and splitting each take $d$ rounds of error correction to stabilize the surface,
\begin{equation}
C_{CNOT}=C_{prep}+4d=5d
\end{equation}
and 
\begin{eqnarray}
V_{CNOT}&=&V_{prep}+(3n+2(2d-1))(C_{prep}+4d)\nonumber \\
&=& 6 d - 44 d^2 + 64 d^3.
\end{eqnarray}

The Hadamard gate is also implemented by the lattice surgery. 
In the lattice surgery technique, firstly Hadamard gates are applied transversally. 
To correct the orientation of the boundary, additional qubits are merged to the boundary and some qubits are split out so that it restores the original boundary orientation. The protocol ends with moving the surface back to the original position. 
It takes $d$ cycles to stabilize the original surface after applying transversal $H$, $d$ cycles for lattice merger, $d$ cycles for lattice splitting, and $d$ cycles for SWAP operations to move the lattice back to the original position.  Thus, $C_H=4d$. 
For circuit volume, we need a bigger surface to carry out merger and split by one more column and row of qubits. Thus, $V_H=(2d)^2 C_H=16d^3$.

For implementing phase gate, we use the circuit in Fig.\fig{S_synthesis}. 
A good thing about this circuit is that the ancilla state $\ket{S}=S\ket{+}=(\ket{0}+i\ket{1})/\sqrt{2}$ is preserved.
Thus, once a purified $\ket{S}$ state is prepared at the beginning of the computation, it can be reused whenever a phase gate needs to be applied. After averaging over a whole computation, the volume use for the distillation process at the beginning will be negligible per one logical gate construction. Note that if only local interactions are allowed, it may take additional circuit volume when the qubit to which the phase gate should be applied is far from the stored $\ket{S}$ state. 
Thus, our estimation should be considered as a lower bound of the actual circuit volume under the setting in which only local interactions are allowed.  
It gives $C_S=2C_{CNOT}+2C_H=18d$ and $V_S=2V_{CNOT}+2V_H+2nC_H=20 d - 120 d^2 + 192 d^3$.

Combining these building blocks, we find the number of cycles and volume required to implement a $T$ gate and a Toffoli gate.

For distilling a $T$ state, $\ket{T}=T\ket{+}$, we use the circuit in \cite{Fowler2012a} which takes 15 $\ket{T}$ states and output 1 $\ket{T}$ with lower error rate.  
It takes 7 surface code cycles for CNOTs and 2 steps for transversal $T$ and measurements, which is 1/4 surface code cycle. Ignoring the last 1/4 cycles, we get $C_{\ket{T}}=7C_{CNOT}=35d$. 
With some parallelization, we get $V_{\ket{T}}=16V_{prep}+V_{CNOT}+7(5V_{CNOT}+6nC_{CNOT})+\frac{1}{4}\cdot16n d=446 d - 2504 d^2 + 3224 d^3$

For implementing a $T$ gate, we use the usual gate teleportation technique \cite{Nielsen2010}.
The $S$ gate correction is applied with probability $1/2$.
We get $C_T = C_{CNOT}+\frac{1}{2}C_S=14d$ and $V_T =V_{\ket{T}}+V_{CNOT}+\frac{1}{2}V_S=462 d - 2608 d^2 + 3384 d^3$.
Since the surface code is CSS, we can transversally make measurements on all the data qubits, and extract eigenvalue for measurement operator. 
Thus, measurement is done with only one time step, which is 1/8 of one surface code cycle, we ignore the volume due to the measurement.

Toffoli gate synthesis in \cite{Jones2013} consists of two steps. 
In the first part, one constructs the $\mbox{Toffoli}^*$ gate, which is Toffoli gate followed by controlled-$S^{\dagger}$ gate, where four $T$ gates and two $H$ gates are used. Also, note that one logical ancilla block is used. 
The second part takes the $\mbox{Toffoli}^*$ gate to the usual Toffoli gate with help of one additional ancilla block. 
By construction of the synthesis circuit, we get
\begin{equation}
C_{\mbox{Toffoli}^*}=2C_H+4C_{CNOT}+C_T=42d
\end{equation}
and
\begin{eqnarray}
V_{\mbox{Toffoli}^*} &=&V_{prep}+6n C_H+2V_H+8V_{CNOT}+4 V_T\nonumber \\
&=&1921 d - 10884 d^2 + 14180 d^3
\end{eqnarray}
Second part of the circuit gives
\begin{equation}
C_{\mbox{Toffoli}}=C_{\mbox{Toffoli}^*}+C_{S}+C_{CNOT}+C_{H}+C_{\ket{T}}=104d
\end{equation}
and 
\begin{eqnarray}
V_{\mbox{Toffoli}}&=& V_{\mbox{Toffoli}^*}+V_{prep}+nC_{\mbox{Toffoli}^*}+V_S+V_{CNOT}\nonumber \\
&+&V_H +3n(C_S+C_{CNOT}+C_H)\nonumber \\
&+&n(C_{\mbox{Toffoli}^*}+C_{CNOT}+C_H)\nonumber \\
&=& 2118 d - 11732 d^2 + 15136 d^3
\label{eq:VTof_1st}
\end{eqnarray}
where the unit for the volume is [qubit $\cdot$ cycle].
We included $C_{\ket{T}}$ in $C_{\mbox{Toffoli}}$ because cycles in the distillation circuit also contribute increasing in the final logical error rate. 
Note that it includes the ancilla qubits for keeping $\ket{S}$ state that is kept during the whole computation. 

Another way to implement logical Toffoli gate on the surface code is to use Toffoli state. 
To locally prepare the Toffoli state, we use the protocol that takes eight $\ket{H}$ states and outputs one Toffoli state \cite{Eastin2013}.
In the preparation circuit, there are two $Y(\pi/4)$ gates and four $Y(-\pi/4)$ gates, which are rotations with respect to $Y$ axis. 
These gates are implemented using $\ket{H}$ state with $Y$ basis measurement, controlled-$Y$ gate, and $Y(\pm \pi/2)$ gate. To implement these gates on the surface code, we use phase gates and Hadamard gates to rotate them to $X$ basis measurement, CNOT gate, and phase gate. 
We then obtain 
\begin{eqnarray}
C_{Y(\pi/4)}&=&C_S+C_{CNOT}+C_S+(2C_H+C_S)/2\nonumber \\
&=&54d
\end{eqnarray}
and 
\begin{eqnarray}
V_{Y(\pi/4)}&=& V_{prep}+V_S+V_{CNOT}+2V_S+(2V_H+V_S)/2\nonumber \\
&=&77 d - 468 d^2 + 756 d^3
\end{eqnarray}
Using these, we obtain 
\begin{eqnarray}
C_{\ket{\mbox{Toffoli}}}&=&7C_{CNOT} + (15/2)C_S + 3C_H\nonumber \\
&=&182d
\end{eqnarray} 
and 
\begin{eqnarray}
V_{\ket{\mbox{Toffoli}}}&=&4V_{prep} + 3(2V_{CNOT} + 2 V_{Y(\pi/4)} + 2nC_{Y(\pi/4)})\nonumber \\ 
&&+ V_{CNOT} + 2nC_{CNOT}\nonumber \\
&=&842 d - 4468 d^2 + 6336 d^3
\end{eqnarray}
where $\ket{\mbox{Toffoli}}$ refers to the Toffoli state.
Cycles and volume for the teleportation circuit, which we write $C_{tele}$ and $V_{tele}$ are
\begin{eqnarray}
C_{tele}&=&C_{CNOT}+1/2(3C_{CNOT}+2C_H)\nonumber\\
&=&16.5d
\end{eqnarray}
\begin{eqnarray}
V_{tele}&=&3V_{CNOT}+\{4nC_H+ 3(V_{CNOT}+nC_{CNOT})\}/2 \nonumber \\
&=&42.5 d - 260 d^2 + 350 d^3
\end{eqnarray}
Combining all of them, we get 
\begin{eqnarray}
C_{\mbox{Toffoli}}=198.5d
\end{eqnarray}
and
\begin{eqnarray}
V_{\mbox{Toffoli}}=7076 d - 37824 d^2 + 53488 d^3
\label{eq:VTof_2nd}
\end{eqnarray}
where the unit for the volume is [qubit $\cdot$ cycle].
Fig.~\fig{volume_Tofstate_T} shows circuit volume with unit [qubit$\cdot$step] in terms of code distance for both ways of implementation. 
We can see that the the scheme with Toffoli state has lower circuit volume.
It is the reason why the scheme with Toffoli state is discussed in the main text.
 
 \begin{figure}[htbp] 
\begin{center}
\includegraphics[scale=1.1]{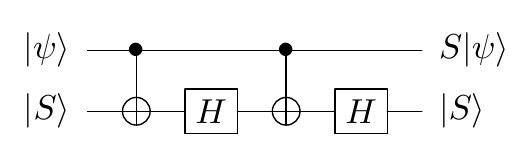}
\caption{Circuit identity used for implementing $S$ gate. Here $\ket{S}=S\ket{+}=(\ket{0}+i\ket{1})/\sqrt{2}$.}
\label{fig:S_synthesis}
\end{center}  
\end{figure}

 \begin{figure}[htbp] 
\begin{center}
\includegraphics[scale=0.3]{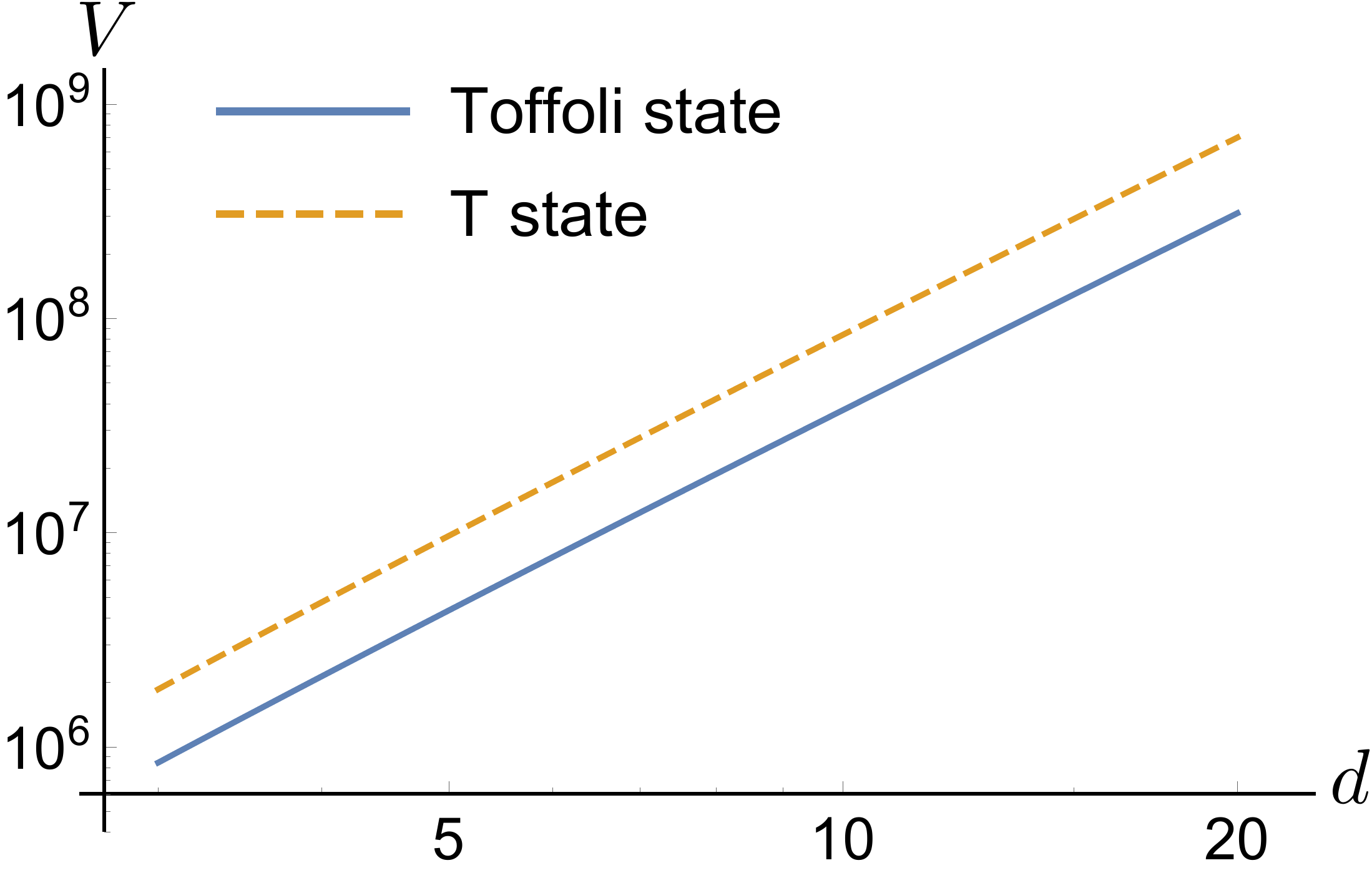}
\caption{Circuit volume for two different implementations of Toffoli gate. Dashed: gate synthesis using $T$ gate. Solid: Toffoli state scheme}
\label{fig:volume_Tofstate_T}
\end{center}  
\end{figure}

\clearpage

\bibliographystyle{apsrev4-1}

%

\end{document}